\documentclass[onecolumn,journal]{IEEEtran}

\usepackage{amsmath, amssymb, amsthm, amsfonts}
\usepackage[pdftex, bookmarks=true]{hyperref}
\usepackage{url}
\usepackage{cmap} 
\usepackage{bm}  
\usepackage{mathtools} 
\usepackage{graphicx}
\usepackage{xcolor}
\usepackage{mdframed}

\newtheorem{theorem}{Theorem}
\newtheorem{proposition}[theorem]{Proposition}
\newtheorem{corollary}[theorem]{Corollary}
\newtheorem{lemma}[theorem]{Lemma}
\newtheorem*{claim*}{Claim}

\theoremstyle{definition}
\newtheorem{remark}[theorem]{Remark}
\newtheorem*{remark*}{Remark}

\newtheorem{definition}[theorem]{Definition}
\newtheorem{problem}[theorem]{Problem}
\newtheorem*{problem*}{Problem}
\newtheorem{example}[theorem]{Example}
\newtheorem*{example*}{Example}

\def\ds{\displaystyle}

\def\P{\mathbb{P}}  

\def\Rb{\mathbb{R}}

\def\Jc{\mathcal{J}}
\def\Jcor{\mathcal{J}_{\textrm{or}}}

\def\Oc{\mathcal{O}}
\def\Pc{\mathcal{P}}  

\def\Xc{\mathcal{X}}
\def\Yc{\mathcal{Y}}
\def\Zc{\mathcal{Z}}
\def\qb{\bm{q}}
\def\qprb{\bm{q'}}
\def\rb{\bm{r}}
\def\rprb{\bm{r'}}
\def\rhatb{\bm{\hat{r}}}
\def\ab{\bm{a}}
\def\aprb{\bm{a'}}
\def\bb{\bm{b}}
\def\pb{\bm{p}}
\def\tb{\bm{t}}
\def\ub{\bm{u}}
\def\vb{\bm{v}}
\def\eb{\bm{e}}
\def\rhob{\bm{\varrho}}
\def\alb{\bm{\alpha}}
\def\jh{\hat{\jmath}}

\def\Kq{K_{\textrm{q}}}
\def\Kr{K_{\textrm{r}}}
\def\Krst{K^\ast_{\textrm{r}}}
\def\Ka{K_{\textrm{a}}}
\def\La{L}

\def\Fr{F_{\textrm{r}}}

\def\phiq{\varphi_{\textrm{q}}}
\def\phir{\varphi_{\textrm{r}}}
\def\phia{\varphi_{\textrm{a}}}
\def\phiaX{\varphi_{\textrm{a}}^X}
\def\Qto{\stackrel{Q}{\longmapsto}}
\def\Rto{\stackrel{R}{\longmapsto}}

\def\epsact{\eps_{\mbox{\scriptsize{act}}}}
\def\epsprec{\eps_{\mbox{\scriptsize{prec}}}}

\def\al{\alpha}
\def\be{\beta}
\def\eps{\varepsilon}
\def\la{\lambda}
\def\lam{\lambda_{\min}}

\def\sm{\setminus}

\DeclareMathOperator{\supp}{supp}
\DeclareMathOperator{\ints}{int}

\def\gequp{\rotatebox[origin=c]{270}{$\geqslant$}}

\DeclareMathOperator{\VP}{VP}
\DeclareMathOperator{\THB}{TH}

\DeclarePairedDelimiterX{\dklx}[2]{(}{)}{#1\;\delimsize\|\;#2}
\newcommand{\dkl}{D_{\mathrm{KL}}\dklx}

\newcommand{\defeq}{\mathrel{\vcenter{\baselineskip0.5ex \lineskiplimit0pt
                     \hbox{\scriptsize.}\hbox{\scriptsize.}}}%
                     =}

\begin{document}

\title{Conditional graph entropy\\ as an alternating minimization problem}

\author{
    \IEEEauthorblockN{Viktor Harangi\IEEEauthorrefmark{1}, Xueyan Niu\IEEEauthorrefmark{2}, Bo Bai\IEEEauthorrefmark{2}}\\
    \IEEEauthorblockA{\IEEEauthorrefmark{1}Alfr\'ed R\'enyi Institute of Mathematics, Budapest, Hungary
    \\harangi@renyi.hu}\\
    \IEEEauthorblockA{\IEEEauthorrefmark{2}Theory Lab, Central Research Institute, 2012 Labs, Huawei Technologies Co. Ltd., Hong Kong SAR, China\\
    \{niuxueyan3,baibo8\}@huawei.com}
    \thanks{During this project the first author received partial support from NRDI 
(grant KKP 138270) and from the Hungarian Academy of Sciences (J\'anos Bolyai Scholarship).}
}

\maketitle


\begin{abstract}
Conditional graph entropy is known to be the minimal rate for a natural functional compression problem with side information at the receiver. In this paper we show that it can be formulated as an alternating minimization problem, which gives rise to a simple iterative algorithm for numerically computing (conditional) graph entropy. This also leads to a new formula which shows that conditional graph entropy is part of a more general framework: the solution of an optimization problem over a convex corner. In the special case of graph entropy (i.e., unconditioned version) this was known due to Csisz\'ar, K\"orner, Lov\'asz, Marton, and Simonyi. In that case the role of the convex corner was played by the so-called vertex packing polytope. In the conditional version it is a more intricate convex body but the function to minimize is the same. Furthermore, we describe a dual problem that leads to an optimality check and an error bound for the iterative algorithm.
\end{abstract}

\section{Introduction}

We 
consider the problem of computing \emph{conditional graph entropy}. Orlitsky and Roche \cite{orlitsky_roche} used this entropy notion to characterize the optimal rate of lossless functional compression with side information at the decoder. Despite playing an inherent role in data compression, little can be found in the literature about conditional graph entropy. 

\subsection{Background and related work}

\subsubsection*{Conditional graph entropy}
Suppose that the random variable $X$ takes values in a finite alphabet $\Xc$, where not every pair of letters can be distinguished. Let $G$ be a graph with vertex set $V(G)=\Xc$ describing which pairs are distinguishable: $x,x' \in \Xc$ can be distinguished if and only if $xx'$ is an edge of $G$. Furthermore, we say that the sequences $x_1,\ldots,x_\ell$ and $x'_1,\ldots,x'_\ell$ are distinguishable if $x_i$ and $x'_i$ are distinguishable for at least one index $i$. We wish to encode an i.i.d.\ sequence $X_1, \ldots, X_\ell$ with high probability in a way that distinguishable sequences are mapped to different codewords. An \emph{independent set} of $G$ contains no edges, and hence any two letters in the set are indistinguishable. Therefore one possible strategy is to replace each $X_i$ with an independent set $J_i \subseteq \Xc$ containing $X_i$, and encode the sequence $J_1,\ldots,J_\ell$ instead. If we do this randomly in a way that $(X_i,J_i)$ are i.i.d.\ samples of some $(X,J)$ where $J$ is a random independent set containing $X$, then we can encode the $J_i$ sequence with rate $H(J)$ (asymptotically as $\ell \to \infty$). Note that the number of times any given typical $X_i$ sequence is covered has exponential rate $H(J|X)$. Based on this, one can design an encoding with rate $H(J) - H(J|X) = I(X;J)$. Then, for a given $X$, one needs to choose $(X,J)$ in a way that the mutual information $I(X;J)$ is as small as possible. K\"orner showed that this is the best achievable code rate and introduced the corresponding notion of \emph{graph entropy} \cite{korner1973}:
\begin{equation} \label{eq:mut_inf}
H_G(X) = \min_{J} I(X;J), \quad \mbox{where $J$ is a random independent set of $G$ such that $X \in J$.}
\end{equation}

The analogous problem with side information $Y_i$ at the receiver leads to the notion of \emph{conditional graph entropy} $H_G(X|Y)$. Let $(X,Y)$ be discrete random variables of some given joint distribution and let $(X_i,Y_i)$ be i.i.d.\ samples. We assume that the decoder knows the sequence $Y_1, Y_2, \ldots$. If we want to use the same approach (i.e., choosing a random $J$), then $J$ and $Y$ should be independent conditioned on $X$ (because the sender does not know $Y_i$ when choosing $J_i$). This can be made rigorous, leading to the following formula:
%
\begin{equation} \label{eq:mut_inf_cond}
H_G(X|Y) = \min_J I(X;J \, | \, Y) 
= \min_J \bigg( H(J | Y) 
- \underbrace{H(J|X,Y)}_{H(J|X)} \bigg),
\end{equation}
%
where $J$ is a random independent set of $G$ such that $X \in J$, and $J$ and $Y$ are conditionally independent conditioned on $X$ (which is equivalent to saying that $Y - X - J$ is a Markov chain).

A lot of work has been done regarding graph entropy since K\"orner \cite{korner1973} introduced the notion in 1973; see the surveys \cite{survey,survey2}. In particular, Csisz\'ar, K\"orner, Lov\'asz, Marton, and Simonyi \cite{entropy_splitting} found a new way to express graph entropy based on a beautiful connection to the so-called \emph{vertex packing polytope} $\VP(G)$, leading to, among other things, an elegant information theoretic characterization of \emph{perfect graphs}. This connection motivated the study of a more general framework, namely, \emph{entropy functions} corresponding to \emph{convex corners}. (See \cite[Proposition 5.4]{Vrana2021} for a recent characterization of such functions.) Besides $\VP(G)$, another notable convex corner associated to graphs is the \emph{theta body} $\THB(G)$ defined by Gr\"otschel, Lov\'asz, and Schrijver \cite{grotschel1986}. It is closely related to the \emph{Lov\'asz number} (or $\vartheta$ function), originally introduced in \cite{lovasz1979} for bounding the Shannon capacity of a graph. 

Much less is known about conditional graph entropy, however. Let us first describe its connection to functional compression.

\subsubsection*{Compression with side information}
Suppose now that the receiver wishes to recover the values $f(X_i,Y_i)$ of some given function $f \colon \Xc \times \Yc \to \Zc$ (with high probability, over long blocks) as depicted in Figure \ref{fig:sideinfo}. Orlitsky and Roche \cite{orlitsky_roche} showed that the minimal rate of information that needs to be transmitted is precisely the conditional graph entropy of the so-called \emph{characteristic graph}, which is defined on the vertex set $\Xc$ as follows: vertices $x_1, x_2 \in \Xc$ are connected with an edge if and only if 
\[
\exists y\in \mathcal{Y} \mbox{ s.t.\ } \big( f(x_1,y) \neq f(x_2,y) \,\, 
\&\,\, \P(X=x_1,Y=y)>0 \,\,\&\,\, \P(X=x_2,Y=y) > 0 \big).    
\]
(This definition goes back to Witsenhausen \cite{Witsenhausen1976}.) We mention that in the special case $f(x,y)=x$, which was already studied in Shannon's classical work \cite{ShannonWeaver49}, the optimal rate is given by the conditional entropy $H(X|Y)=H(X,Y)-H(Y)$. 
\begin{figure}[ht]
\centering
\includegraphics[width=3in]{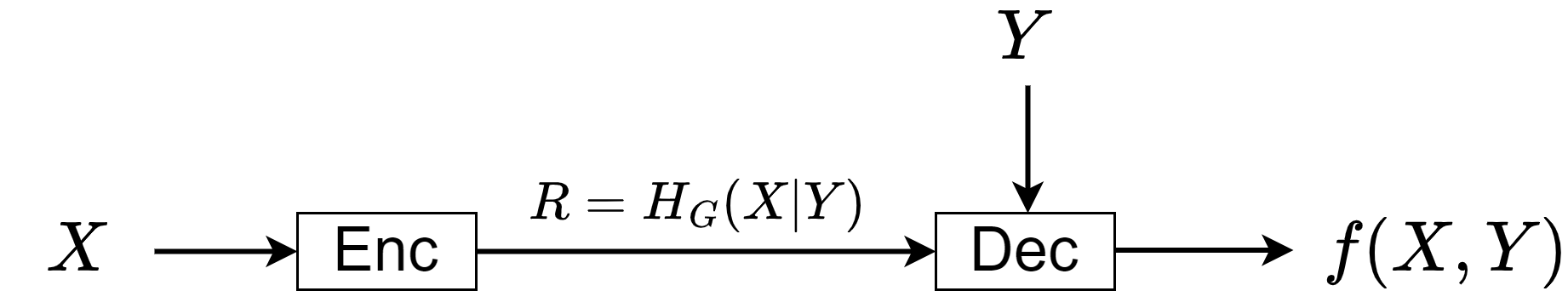}
\caption{The problem of functional compression with side information at the decoder}
\label{fig:sideinfo}
\end{figure}

Also note that these problems are naturally connected to graph coloring. Doshi et al.\ \cite{doshi_et_al} extended the notion of \emph{chromatic entropy} \cite{Alon1996} and defined \emph{conditional chromatic entropy}. They showed that first coloring a sufficiently large power graph 
and then encoding the colors achieves conditional graph entropy. 

\subsubsection*{Alternating optimization}
As we have mentioned, it turns out that conditional graph entropy can be obtained by alternating optimization. This family includes a large number of problems from a variety of fields. A usual feature is that although the optimum has no closed-form expression, there is an efficient way to optimize in each variable. Thus, optimizing in the different variables in turns can lead to a good numerical approximation of the optimum. A well-known example is the expectation–-maximization (EM) algorithm.  Another prominent example is the Blahut--Arimoto (BA) algorithm \cite{Arimoto1972,Blahut1972}, which deals with the capacity of discrete memoryless channels. To find the capacity achieving input, the algorithm turns the objective into a double supremum and alternately optimizes over the distribution parameters from random initialization. Also, \cite{Dupuis2004} provides a numerical method for the Gel'fand-Pinsker problem \cite{GEL1980}, where noncausal state information is known at the encoder. Generalization of the BA algorithm for finite-state channels, proposed in \cite{Vontobel2008}, also takes advantage of the iterative nature of alternating optimization. To the best of our knowledge, no similar procedure was proposed for the Orlitsky--Roche problem \cite{orlitsky_roche} beforehand.

It is important to point out that Csisz\'ar and Tusn\'ady \cite{csiszar_tusnady} initiated the systematic study of such problems in the 1980s already. The cornerstone of their theory is a collection of inequalities called 3-point, 4-point, and 5-point properties. They will play a key role in our problem as well.

\subsection{Notations} \label{sec:notations}
Random variables are denoted by uppercase letters ($X$, $Y$, $J$), while their realizations are denoted by lowercase letters ($x$,$y$,$j$). The (discrete) alphabet of a random variable is denoted by the corresponding script letter ($\Xc$,$\Yc$,$\Jc$). For brevity, we write $\sum_x$ for $\sum_{x \in \Xc}$, and $\sum_y$ for $\sum_{y \in \Yc}$. We use $\P(\cdot)$ to denote the probability of an event, and the following shorthand notations will be used as well: 
\begin{align*} 
p_{x,y} &\defeq \P(X=x,Y=y) ; \\
p_x &\defeq \P(X=x)=\textstyle\sum_y p_{x,y} ; \\
p^y &\defeq \P(Y=y)=\textstyle\sum_x p_{x,y} ; \\
p_{x|y} &\defeq \P(X=x|Y=y)=p_{x,y}/p^y ; \\
p^{y|x} &\defeq \P(Y=y|X=x)=p_{x,y}/p_x . 
\end{align*}

In most settings $j$ denotes a subset of $\Xc$. When a graph $G$ is given on the vertex set $\Xc$, then $j$ always denotes an independent set: $j \subseteq \Xc$ is a set of vertices such that the induced subgraph $G[j]$ contains no edge. In this setting $J$ stands for a random independent set. Then $\Jc$ denotes the set of all independent sets, while $\Jc_x$ is the set of independent sets containing $x$.

\subsection{Contributions} \label{sec:contributions}

\subsubsection*{Alternating minimization}

Let us consider the following optimization problem.
\begin{problem*}
Suppose that we have two finite families of probability measures on a given finite set $\Jc$: $\mu_x$, $x \in \Xc$ and $\nu_y$, $y \in \Yc$. In the first family for each $x \in \Xc$ we have a constraint: the support $\supp \mu_x$ must be contained in a given subset $\Jc_x$ of $\Jc$. Find the measures $\mu_x$, $\nu_y$ that minimize the weighted sum of the Kullback--Leibler divergences: 
\[ \sum_{x,y} p_{x,y} \dkl{\mu_x}{\nu_y} 
\mbox{ for some given weights } p_{x,y} \geq 0 . \]
That is, given $\Jc_x \subseteq \Jc$, $x \in \Xc$ and $p_{x,y} \geq 0$, $x \in \Xc, y \in \Yc$, find the minimum of the above sum under the constraint $\supp \mu_x \subseteq \Jc_x$. 
\end{problem*}

In our setting we have random variables $X$ and $Y$ taking values in the finite sets $\Xc$ and $\Yc$, respectively, and $G$ is a graph on the vertex set $\Xc$. Then by $j$ we denote an independent set of $G$, hence each $j$ is a subset of $\Xc$. We choose $\Jc$ to be the set of all $j$, while 
\[ \Jc_x \defeq \{ j \, : \, x \in j \} \]
consists of the independent sets containing a fixed $x$. With this setup and with $p_{x,y} \defeq \P(X=x,Y=y)$, the minimum of the problem above turns out to be precisely $H_G(X|Y)$.

To get concrete formulas, let us represent the distributions $\mu_x$ and $\nu_y$ by the following vectors: 
\begin{align*} 
\qb &= \big( q_{j|x} \big)_{(j,x) \in \Jc \times \Xc} 
\in \Rb^{\Jc \times \Xc}; \\
\rb &= \big( r_{j|y} \big)_{(j,y) \in \Jc \times \Yc} 
\in \Rb^{\Jc \times \Yc} ,
\end{align*} 
where $q_{j|x}$ and $r_{j|y}$ stand for $\mu_x(\{j\})$ and $\nu_y(\{j\})$, respectively.\footnote{We index the coordinates/variables by $j|x$ and $j|y$ to emphasize the fact that they express certain conditional probabilities, see the proof of Proposition \ref{prop:min_phiq} for details. This notation may also serve as a reminder that $q_{j|x}$ and $r_{j|y}$ have to sum up to $1$ for any fixed $x$ and $y$, respectively.} The constraints for $\qb$ and $\rb$ lead to the following definition.
\begin{definition} \label{def:polytopes}
We define the convex polytopes $\Kq \subset \Rb^{\Jc \times \Xc}$ and $\Kr \subset \Rb^{\Jc \times \Yc}$ as
\[
\Kq \defeq \bigg\{ \qb= \big( q_{j|x} \big) \, : \, 
q_{j|x} \geq 0; \, 
\sum_{j \ni x} q_{j|x} = 1 \, (\forall x\in\Xc); \, 
q_{j|x}=0 \mbox{ if } x \notin j \bigg\} 
\]
and
\[
\Kr \defeq \bigg\{ \rb= \big( r_{j|y} \big) \, : \, 
r_{j|y} \geq 0; \, 
\sum_{j} r_{j|y} = 1 \, (\forall y\in\Yc) \bigg\} .
\]
By $\ints(\Kq)$ and $\ints(\Kr)$ we denote the relative interiors of the polytopes (within their affine hull).
\end{definition}
In the sequel we will always assume that $\qb \in \Kq$ and $\rb \in \Kr$. Then 
\[ \dkl{\mu_x}{\nu_y} = \sum_j q_{j|x} \log \frac{q_{j|x}}{r_{j|y}} .\]
Therefore we need to minimize the function
\begin{equation} \label{eq:phi_intro}
\varphi(\qb,\rb) \defeq \sum_{x,y,j} 
p_{x,y} \, q_{j|x} \log \frac{q_{j|x}}{r_{j|y}} 
\end{equation}
over $\qb \in \Kq$ and $\rb \in \Kr$. 

As we will see, this is an alternating minimization problem. The point is that if we fix one of the two variables $\qb$ and $\rb$, then there are explicit formulas for the optimal choice of the other variable: we will define maps 
\[ Q \colon \Kr \to \Kq \mbox{ and } R \colon \Kq \to \Kr \]
such that $\rb=R(\qb)$ is the optimal choice for a fixed $\qb$, and similarly $\qb=Q(\rb)$ is optimal for a fixed $\rb$; that is, for any $\qb$ and $\rb$ we have 
\[
\varphi( \qb, \rb) \geq \varphi\big( \qb, R(\qb) \big) \mbox{ and }
\varphi( \qb, \rb) \geq \varphi\big( Q(\rb), \rb \big) .
\]
Using $Q$ and $R$ we can explicitly define the following functions: 
\begin{align*}
\phiq (\qb) &\defeq \varphi\big( \qb, R(\qb) \big) = \min_{\rb} \varphi(\qb,\rb) \mbox{ and}\\ 
\phir (\rb) &\defeq \varphi\big( Q(\rb), \rb \big) = \min_{\qb} \varphi(\qb,\rb) ,
\end{align*}
They clearly have the same minimum as $\varphi$. When we work out the details in Section \ref{sec:alt_min_prob}, we will see that the q-problem $\min \phiq$ is actually equivalent to the original formula \eqref{eq:mut_inf_cond} for conditional graph entropy (see Proposition \ref{prop:min_phiq}).
\begin{theorem} \label{thm:cond_graph_entropy}
We have the following formulas for conditional graph entropy:
\[ H_G(X|Y) = \min_{\Kq \times \Kr} \varphi 
= \min_{\Kq} \phiq = \min_{\Kr} \phir .\]
\end{theorem}

\subsubsection*{Algorithm} 

When trying to find the minimum of $\varphi(\qb,\rb)$, the fact that we can easily optimize in either variable (while the other is fixed) gives rise to the following simple iterative algorithm. Let us start from a point $\qb^{(0)}$ and apply $R$ and $Q$ alternately:
\begin{equation} \label{eq:alt_seq}
\qb^{(0)} \Rto \rb^{(0)} \Qto \qb^{(1)} \Rto \rb^{(1)} \Qto 
\qb^{(2)} \Rto \rb^{(2)} 
\cdots .
\end{equation}
The corresponding $\varphi$-value decreases at each step:
\begin{center}
\begin{tabular}{ccc}
$\varphi\big( \qb^{(0)}, \rb^{(0)}\big)$ & $=$ & $\phiq( \qb^{(0)} )$ \\
\gequp &  & \gequp \\
$\varphi\big( \qb^{(1)}, \rb^{(0)}\big)$ & $=$ & $\phir( \rb^{(0)} )$ \\
\gequp &  & \gequp \\
$\varphi\big( \qb^{(1)}, \rb^{(1)}\big)$ & $=$ & $\phiq( \qb^{(1)} )$ \\
\gequp &  & \gequp \\
$\varphi\big( \qb^{(2)}, \rb^{(1)}\big)$ & $=$ & $\phir( \rb^{(1)} )$ \\
\gequp &  & \gequp \\
$\varphi\big( \qb^{(2)}, \rb^{(2)}\big)$ & $=$ & $\phiq( \qb^{(2)} )$ \\
$\vdots$ &  & $\vdots$
\end{tabular} 
\end{center}
%
%
One can also think of this alternating optimization as ``jumping'' between the q-problem $\min_{\Kq} \phiq$ and the r-problem $\min_{\Kr} \phir$ using the maps $Q \colon \Kr \to \Kq$ and $R \colon \Kq \to \Kr$. The value to minimize (i.e., the $\phiq$-value and the $\phir$-value, respectively) always decreases, so with each step we get closer to the optimum.

Following the footsteps of the general theory of Csisz\'ar and Tusn\'ady \cite{csiszar_tusnady}, we will show that, for an arbitrary starting point $\qb^{(0)}$ in the relative interior $\ints(\Kq)$, the iterative process converges to the minimum.
\begin{theorem} \label{thm:convergence}
For an arbitrary starting point $\qb^{(0)} \in \ints(\Kq)$ consider the sequence \eqref{eq:alt_seq} obtained by alternating optimization. Then $\varphi\big( \qb^{(n)}, \rb^{(n)}\big)$ is a decreasing sequence that converges to $\ds \min_{\Kq \times \Kr} \varphi$ 
as $n \to \infty$.
\end{theorem}
%
%
We implemented the algorithm in Python and made the codes publicly available in a GitHub repository \cite{ge_github}. 

\subsubsection*{Convex corners}
As we have mentioned, the q-problem $\min \phiq$ gives back the original formula \eqref{eq:mut_inf_cond}. On the other hand, the r-problem $\min \phir$ gives us a new formula. We make this new formula explicit in the next theorem because it shows how conditional graph entropy is related to convex corners (a known phenomenon in the unconditioned case).
\begin{theorem} \label{thm:new_formula}
For any (maximal) independent set $j$ of $G$ and any possible value $y$ of $Y$ we have a variable $r_{j|y}$. Then conditional graph entropy can be expressed as the solution of the following optimization problem:
\begin{equation} \label{eq:new}
H_G(X|Y) = \min - \sum_x p_x
\log \bigg( \sum_{j \, : \, x \in j} \, \prod_y \big( r_{j|y} \big)^{
p^{y|x}}
\bigg) ,
\end{equation}
where the minimum is taken over all choices of $r_{j|y} \geq 0$ satisfying $\sum_j r_{j|y} = 1$ for each fixed $y$.
\end{theorem}

We can easily turn this new r-problem into another one (that we will call the a-problem) which attests that conditional graph entropy is a special case of a more general entropy notion defined for convex corners.\footnote{A convex corner of $\Rb^\Xc$ is a convex compact set in the positive orthant $[0,\infty)^\Xc$ that is \emph{downward closed}, i.e., if we take any point in the set and decrease some of its coordinates, then the new point still lies in the set (see Definition \ref{def:convex_corner}).} 

To see this connection, note that \eqref{eq:new} is in the form 
\[ H_G(X|Y) = \min_{\Kr} - \sum_x p_x 
\log A_x ,\]
where $A_x$ is a function of the variables $r_{j|y}$. The key property is that $A_x \colon \Kr \to [0,1]$ is a concave function for each $x$. It means that the set of image points $\ab = ( a_x )_{x \in \Xc}$ with $a_x = A_x(\rb)$, as $\rb$ ranges over $\Kr$, (essentially) defines a convex corner $\Ka$ in $\Rb^\Xc$. Then we have 
\[
H_G(X|Y) = \min_{\Ka} \phia 
\mbox{, where } \phia( \ab ) = - \sum_x p_x 
\log a_x .
\]
A nice feature of this a-problem is that the minimum is attained at a single point $\ab \in \Ka$ because $\phia$ is strictly convex (provided that $p_x=\P(X=x)>0$ for each $x$). Also note that $\phia$ depends only on the distribution of $X$, while the convex corner $\Ka$ depends only on the graph $G$ and the conditional distributions $Y \, | \, X=x$ for any given $x$. Thus, the parameters of the problem are, so to say, split between $\phia$ and $\Ka$. 

Moreover, we will define another convex corner, denoted by $\La$, that can be regarded as the \emph{dual problem}. To keep the introduction concise, we will postpone the actual definition of $\La$ and the precise statements until Section \ref{sec:dual}. In short, we will show that  
\[ H_G(X|Y) = \min_{\Ka} \phia = - \min_{\La} \phia ,\] 
and a vector $\ab = \big( a_x \big)_{x \in \Xc} \in \Ka$ is optimal (i.e., the minimum point of $\phia$) if and only if $\ab^{-1} \defeq \big( a_x^{-1} \big)_{x \in \Xc} \in \La$. This provides a fairly simple way to check optimality, and even leading to an error bound for our iterative algorithm as we will explain in Section \ref{sec:algorithm}. 

In the special case when $Y$ is trivial, i.e., $\Yc$ is a one-element set, we get back K\"orner's original setting of graph entropy, and things simplify considerably. For example, $\Ka$ is simply a polytope: the aforementioned \emph{vertex packing polytope} (or \emph{independent/stable set polytope}) $\VP(G)$. 
We did not find any mention in the literature of the fact that graph entropy can be considered as an alternating minimization problem. In particular, to the best of our knowledge, the corresponding iterative algorithm has not been used or proposed before even in this unconditioned setting.

\subsection*{Outline of the paper}
We give more details of our alternating minimization problem in Section \ref{sec:alt_min_prob}, and collect its key properties in Section \ref{sec:convergence}, proving, in particular, the convergence of the iterative process. In Section \ref{sec:corner} we discuss the connection to convex corners and introduce the dual problem. In Section \ref{sec:algorithm} we discuss some details of the iterative algorithm; in particular, an error bound based on the dual problem and a tweak for speeding the convergence up.

\section{The alternating minimization problem} \label{sec:alt_min_prob}

In this section we rigorously introduce the optimization problem $\min \varphi(\qb,\rb)$ described in the introduction. We will use the notations outlined in Section \ref{sec:notations}. 

\subsection{Assumptions}
For the sake of simplicity, we will work under the following three assumptions that do not actually reduce generality. 
\begin{itemize}
\item \textbf{Each $p_x=\textstyle\sum_y p_{x,y}$ and $p^y=\textstyle\sum_x p_{x,y}$ is strictly positive}. (Otherwise we simply delete the corresponding elements from $\Xc$ and $\Yc$.) Note that under this assumption the conditional probabilities $p_{x|y}$ and $p^{y|x}$ all exist.
\item \textbf{The sets $j$ cover $\Xc$}, that is, $\forall x \in \Xc \, \exists j \in \Jc$ s.t.\ $x \in j$. (Otherwise the minimum we will consider would be $\infty$ anyway.) 
\item \textbf{$\Jc$ contains inclusion-wise maximal sets}. (Removing subsets of other sets from $\Jc$ does not change the minimum.)
\end{itemize}

Also, all the results will be true under the more general setting when $\Jc \subseteq \Pc(\Xc)$ is any set of subsets of $\Xc$. That is, the sets $j \in \Jc$  are subsets of $\Xc$ but they do not necessary need to be independent sets of some graph $G$ on $\Xc$. In conclusion, our setup essentially has the following fixed parameters: the probabilities $p_{x,y}$ and a binary relation $\in$ on $\Xc \times \Jc$: whenever $x$ is in the set $j$ we write $x \in j$.\footnote{Equivalently, we may write $j \ni x$. In particular, $\sum_{j \ni x}$ means that the sum runs over sets $j \in \Jc$ containing the (fixed) element $x$.} 

\subsection{The mappings}
Recall the convex polytopes $\Kq$ and $\Kr$ defined in Section \ref{sec:contributions} of the introduction. Now we explicitly define the mappings $Q$ and $R$ between these polytopes along with an auxiliary mapping $A$. In fact, the formula defining $Q$ will make sense only on the subset $\Krst$ where none of the coordinates of $A$ vanishes. We define $Q$ arbitrarily outside $\Krst$. 

\begin{definition} \label{def:AQR}
We define the mappings
$A \colon \Kr \to \Rb^{\Xc}$; 
$Q \colon \Kr \to \Kq \subset \Rb^{\Jc \times \Xc}$; 
$R \colon \Kq \to \Kr \subset \Rb^{\Jc \times \Yc}$
by the following coordinate-wise functions $Q_{j|x}$, $R_{j|y}$, $A_x$: 
\begin{align*}
R_{j|y}(\qb) &\defeq \sum_{x \in j} p_{x|y} \, q_{j|x} ;\\
A_x(\rb) &\defeq \sum_{j \ni x} \, \prod_y \big( r_{j|y} \big)^{p^{y|x}} ;\\
Q_{j|x}(\rb) &\defeq 
\begin{cases}
0  & \mbox{ if } x \notin j ;\\
\prod_y \big( r_{j|y} \big)^{p^{y|x}} \bigg/ A_x(\rb) & \mbox{ if } x \in j .
\end{cases} 
\end{align*}
Since the formula for $Q_{j|x}$ involves a division by $A_x$, it only defines $Q$ over the subset
\begin{equation} \label{def:Krst}
\Krst \defeq \Kr \setminus \bigcup_x A^{-1}_x(0) .    
\end{equation}
For $\rb \in \Kr \setminus \Krst$ let $Q(\rb)$ be an arbitrary point in $\ints(\Kq)$.

In the formulas above we define $t^0=1$ even for $t=0$. This ensures that $A$ is continuous over the entire $\Kr$ even when $p^{y|x}=0$ for some pairs $x,y$. It is also consistent with the convention $0 \cdot \log 0 = 0$ which is implicit in the definition of Shannon entropy.
\end{definition}
It is straightforward to check that $Q(\rb) \in \Kq$ and $R(\qb) \in \Kr$ always hold. 
For example, in the definition of $Q_{j|x}(\rb)$, dividing by $A_x(\rb)$ ensures that their sum is $1$ for any fixed $x$. 

\begin{remark}
Note that $R$ is a linear map and it actually describes how the conditional distributions $J \, | \, Y=y$ can be expressed in terms of $J \, | \, X=x$ in a Markov chain $Y - X - J$; see the proof of Proposition \ref{prop:min_phiq} for details. 
\end{remark}

\subsection{The functions}
Now we can turn our attention to the functions to be minimized. We already gave an explicit formula \eqref{eq:phi_intro} for $\varphi(\qb,\rb)$ in the introduction. However, we did not mention a few subtleties there. In particular, we need to specify the function values when some of the variables $q_{j|x}$ or $r_{j|y}$ are $0$. 
\begin{definition} \label{def:phi}
For $u,v \in [0,1]$ let 
\[ f(u,v) \defeq u \log u - u \log v \]
with the usual conventions $\log 0 = - \infty$ and $0 \cdot \infty = 0$ so that 
\[ f(0,v)=0 \mbox{ if } v \in [0,1] 
\quad \mbox{ and } \quad
f(u,0)=\infty \mbox{ if } u \in (0,1]. \]
Then 
\begin{equation} \label{eq:phi}
\varphi(\qb,\rb) \defeq \sum_{x,y,j} p_{x,y} \, f\big( q_{j|x}, r_{j|y} \big) \end{equation}
is well-defined for any $\qb \in \Rb^{\Jc \times \Xc}$ and $\rb \in \Rb^{\Jc \times \Yc}$. Note that we may restrict the sum for $x \in j$ because otherwise $q_{j|x}=0$, and hence the summand is $0$ anyway.

Let us also define the following auxiliary function that we will need for establishing the so-called 3-point and 4-point properties. 
%
\begin{equation} \label{eq:delta}
\delta(\qb,\qprb) \defeq \sum_{x} p_x \, \sum_{j \ni x} f\big( q_{j|x}, q'_{j|x} \big) 
= \sum_{x} p_x \, \sum_{j \ni x} q_{j|x} \log q_{j|x} - q_{j|x} \log q'_{j|x} .
\end{equation}
%
\end{definition}

One may think of $\varphi(\qb,\rb)$ and $\delta(\qb,\qprb)$ as (non-symmetric) squared distances between these points. We mention that both functions are convex combinations of certain Kullback--Leibler divergences. 
In particular, they are nonnegative and they may be $\infty$. For example, $\varphi(\qb,\rb)=\infty$ if and only if there exist $x,y,j$ such that $r_{j|y}=0$ while $p_{x,y}>0$ and $q_{j|x}>0$. 
It is easy to see that if $\rb \in \Kr \sm \Krst$, then $\varphi(\qb,\rb)=\infty$ for any choice of $\qb \in \Kq$. (The two-line proof of this fact is included in the proof of Proposition \ref{prop:3pt}.)

We include here two useful equivalent formulas for $\varphi$. On the one hand, summing $q_{j|x} \log q_{j|x}$ and $q_{j|x} \log r_{j|y}$ separately gives 
%
\begin{equation} \label{eq:phi_with_R}
\varphi(\qb,\rb) = \sum_x p_x \sum_{j \ni x} q_{j|x} \log q_{j|x} 
- \sum_y p^y \sum_j R_{j|y}(\qb) \log r_{j|y} .
\end{equation}
%
On the other hand, for $\rb \in \Krst$ we can write 
\begin{multline} \label{eq:phi_with_QA}
\varphi(\qb,\rb) = \sum_x p_x \sum_{j \ni x} q_{j|x} 
\bigg( \log q_{j|x} - \sum_y p^{y|x} \log r_{j|y} \bigg) \\
= \sum_x p_x \sum_{j \ni x} q_{j|x} \log 
\frac{q_{j|x}}{\prod_y \big( r_{j|y} \big)^{p^{y|x}}} 
= \sum_x p_x \sum_{j \ni x} q_{j|x} \log 
\frac{q_{j|x}}{Q_{j|x}(\rb) A_x(\rb)} 
\end{multline}
with the remark that if $q_{j|x}$ and $Q_{j|x}(\rb)$ are both $0$, then the fraction in the $\log$ should simply be $1/A_x(\rb)$.

Next we define $\phiq$ and $\phir$. At this point we simply express them using $Q$, $R$, and $\varphi$, but we will shortly see that they are indeed the minimum of $\varphi$ with one of the variables fixed. Using \eqref{eq:phi_with_R} and \eqref{eq:phi_with_QA} we get the following specific formulas: for $\qb \in \Kq$ and $\rb \in \Kr$ let 
%
\begin{equation} \label{eq:phiq}
\phiq(\qb) \defeq \varphi\big(\qb,R(\qb)\big) 
= \sum_x p_x \sum_{j \ni x} q_{j|x} \log q_{j|x} 
- \sum_y p^y \sum_j R_{j|y}(\qb) \log R_{j|y}(\qb) ;
\end{equation}
%
\begin{equation} \label{eq:phir}
\phir(\rb) \defeq \varphi\big(Q(\rb),\rb\big) 
= - \sum_x p_x \log A_x(\rb) 
=  - \sum_x p_x \log \sum_{j \ni x} \, \prod_y \big( r_{j|y} \big)^{p^{y|x}} .  
\end{equation}
%
Note that \eqref{eq:phir} works even for $\rb \notin \Krst$ as all the expressions are $\infty$ in that case.

\begin{proposition} \label{prop:min_phiq}
If $\Jc$ is the set of independent sets of some graph $G$ on the vertex set $\Xc$, then 
\[ H_G(X|Y) = \min_{\Kq} \phiq .\]
\end{proposition}
\begin{proof}
Recall that the original formula \eqref{eq:mut_inf_cond} for $H_G(X|Y)$ involves minimization over random $J$ containing $X$ and independent from $Y$ when conditioned on $X$ (in other words, $Y-X-J$ is a Markov chain). To define such a $J$ one needs to specify the conditional probabilities $\P(J=j|X=x)$ whenever $x \in j$. These conditional probabilities can be represented by a vector $\qb \in \Kq$. Due to the conditional independence, we have the expression 
\begin{equation} \label{eq:cond_prob_expr}
\P(J=j | Y=y) = \sum_{x \in j} p_{x|y} \, \P(J=j | X=x) .
\end{equation}
Note that we defined $R$ using the same linear combinations, see Definition \ref{def:AQR}. Consequently, if the $\P(J=j|X=x)$'s are represented by $\qb$, then the $\P(J=j|Y=y)$'s are represented by $\rb=R(\qb)$. Therefore
\begin{align*}
H(J|X) &= - \sum_x p_x \sum_{j \ni x} q_{j|x} \log q_{j|x} ;\\ 
H(J|Y) &= - \sum_y p^y \sum_j R_{j|y}(\qb) \log R_{j|y}(\qb) ,
\end{align*}
and hence the conditional mutual information $I(X;J \, | \, Y) = H(J|Y)-H(J|X)$ is precisely $\phiq(\qb)$ according to \eqref{eq:phiq}, proving $H_G(X|Y)=\min_{\Kq} \phiq$. 
\end{proof}

\section{Convergence} \label{sec:convergence}

In this section we derive various properties of the the minimization problems introduced in Section \ref{sec:alt_min_prob}. They will culminate in the proof that alternating optimization converges to the true minimum (Theorem \ref{thm:convergence}). We will also prove Theorems \ref{thm:cond_graph_entropy} and \ref{thm:new_formula} along the way. 
\begin{proposition} \label{prop:phi_delta}
The functions $\varphi$ and $\delta$ are nonnegative, lower semicontinuous, and convex. Moreover, $\delta(\qb,\qprb)=0$ if and only if $\qb=\qprb$.
\end{proposition}
\begin{proof}
Recall that $\varphi$ and $\delta$ were defined using the function $f \colon [0,1]^2 \to (-\infty,\infty]$ in Definition \ref{def:phi}. It is well known and easy to show that $f$ is convex and lower semicontinuous, and hence so are $\varphi$ and $\delta$.  

Using the convexity of $f$ we get that for any fixed $x,y$: 
\[
\mbox{both } \sum_{j} f\big( q_{j|x} , r_{j|y} \big) 
\mbox{ and } \sum_{j} f\big( q_{j|x} , q'_{j|x} \big) 
\geq |\Jc| \, f\big( 1/|\Jc| , 1/|\Jc| \big) = 0 
\]
showing that $\varphi, \delta \geq 0$. (This, of course, also follows from their representations as the sum of Kullback--Leibler divergences.) 
\end{proof}

Note that lower semicontinuity implies that $\varphi$ attains its minimum over any compact set. In particular, it has a minimum over $\Kq \times \Kr$. 

\begin{proposition}[$\rb=R(\qb)$ is optimal for fixed $\qb$] \label{prop:R_ineq}
We have 
\[ \varphi(\qb,\rb) \geq \varphi\big( \qb, R(\qb) \big) = \phiq(\qb)  
\mbox{ for any } \qb \in \Kq; \rb \in \Kr . \] 
Equality holds if and only if $\rb=R(\qb)$. 
\end{proposition}
\begin{proof}
Using formula \eqref{eq:phi_with_R} for a fixed $\qb$, it immediately follows from Gibbs' inequality (applied for each $y$ in the second sum) that the unique optimal choice for $\rb$ is $R(\qb)$.
\end{proof}
\begin{proposition}[$\qb=Q(\rb)$ is optimal for fixed $\rb$] \label{prop:Q_ineq}
We have 
\[ \varphi(\qb,\rb) \geq \varphi\big( Q(\rb), \rb \big) = \phir(\rb) 
\mbox{ for any } \qb \in \Kq; \rb \in \Kr . \] 
If $\rb \notin \Krst$, then both sides are $\infty$. Furthermore, for $\rb \in \Krst$ equality holds if and only if $\qb=Q(\rb)$. 
\end{proposition}
\begin{proof}
This is an immediate consequence of the 3-point property (that we will shortly state in Proposition \ref{prop:3pt}) and the fact that $\delta \geq 0$. 
\end{proof}
\begin{corollary} \label{cor:same_min}
Propositions \ref{prop:R_ineq} and \ref{prop:Q_ineq} clearly show that 
\[ \phiq(\qb) = \min_{\rb \in \Kr} \varphi(\qb,\rb) \mbox{ and } 
\phir(\rb) = \min_{\qb \in \Kq} \varphi(\qb,\rb) .\]
In particular, $\varphi$, $\phiq$, $\phir$ have the same minimum over their respective convex domains:  
\[ 
\min_{\Kq \times \Kr} \varphi = \min_{\Kq} \phiq = \min_{\Kr} \phir .\]
It also follows that both $\phiq$ and $\phir$ are convex as they can be obtained as minimizing the convex $\varphi(\qb,\rb)$ in one of the variables.
\end{corollary}

Note that, combined with Proposition \ref{prop:min_phiq}, this completes the proof of Theorem \ref{thm:cond_graph_entropy}. Moreover, Theorem \ref{thm:new_formula} also follows as we simply need to substitute \eqref{eq:phir}, which expresses $\phir$, into $H_G(X|Y)=\min_{\Kr} \phir$. 

From this point on we follow the footsteps of the general theory \cite{csiszar_tusnady} of alternating minimization problems by proving the so-called \emph{3-point and 4-point properties}, and show how they imply convergence to the minimum through the 5-point property. 

The following identity can be thought of as a Pythagorean theorem for the ``squared distances'' $\varphi$ and $\delta$. Csisz\'ar and Tusn\'ady refer to it as the 3-point property. (In their general setting it may hold only as an inequality $\geq$ but in our case we always have equality.)
\begin{proposition}[3-point property] \label{prop:3pt}
For any $\qb \in \Kq$ and $\rb \in \Kr$ we have 
\[ \varphi(\qb,\rb) = \delta\big( \qb, Q(\rb) \big) + 
\underbrace{\varphi\big( Q(\rb), \rb \big)}_{=\phir(\rb)} . \]
%
\begin{center}
\includegraphics[width=5cm]{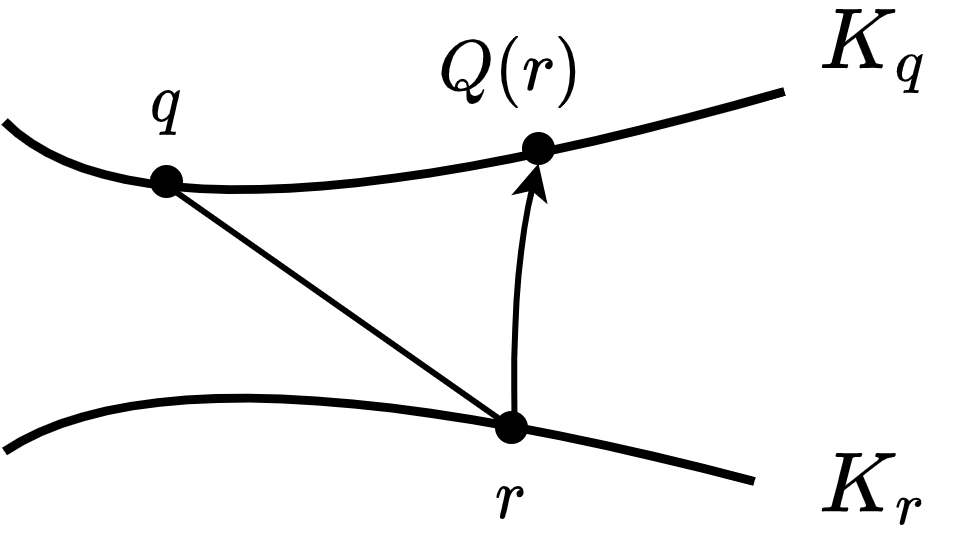}
\end{center}
%
\end{proposition}
\begin{proof}
First let us consider the cases when one of the terms on the right-hand side is $\infty$. In both cases we need to show that one can find $x,y,j$ with $p_{x,y}>0$, $q_{j|x}>0$, $r_{j|y}=0$ so that we can conclude that the left-hand side $\varphi(\qb,\rb)$ is also $\infty$.
\begin{itemize}
\item We have $\phir(\rb)=\infty$ if and only $A_x(\rb)=0$ for some $x$. Fix such an $x$ and take a $j \ni x$ with $q_{j|x}>0$, which must exist as their sum is $1$. Since $A_{x}(\rb)=0$, there must exist $y$ such that $r_{j|y}=0$ and $p_{x,y}>0$. 
\item We have $\delta\big( \qb, Q(\rb) \big)=\infty$ if and only if there exist $j,x$ such that $q_{j|x}>0$ but $Q_{j|x}(\rb)=0$, which means, by the definition of $Q_{j|x}$, that there exists $y$ such that $p_{x,y}>0$ and $r_{j|y}=0$.
\end{itemize}
Otherwise we can simply combine formula \eqref{eq:phi_with_QA} for $\varphi(\qb,\rb)$, formula \eqref{eq:phir} for $\phir(\rb)$, and formula \eqref{eq:delta} for $\delta(\qb,\qprb)$ with $\qprb=Q(\rb)$ to get the claim.
\end{proof}
\begin{proposition}[4-point property] \label{prop:4pt}
For any $\qb,\qprb \in \Kq$ and $\rb \in \Kr$ we have 
\[ \varphi\big( \qb,R(\qprb) \big) \leq \varphi(\qb,\rb) + \delta(\qb,\qprb) .\]
%
\begin{center}
\includegraphics[width=5cm]{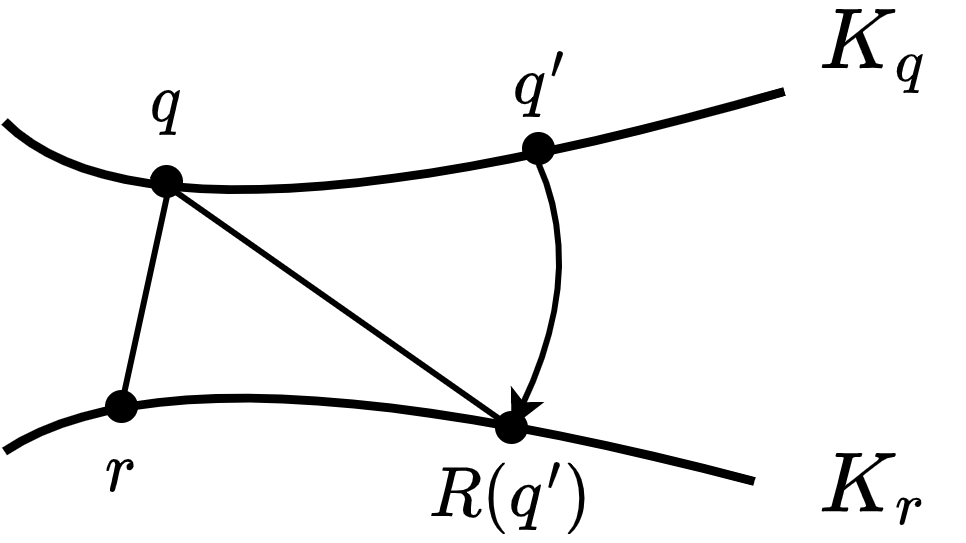}
\end{center}
%
\end{proposition}
\begin{proof}
We may assume that the right-hand side is finite, otherwise the inequality is trivial. It follows that for any triple $x,y,j$ with $p_{x,y}>0$, $q_{j|x}>0$ we must have both $r_{j|y}>0$ and $q'_{j|x}>0$. Let $\rprb \defeq R(\qprb)$. Then $r'_{j|y} \geq p_{x|y} \, q'_{j|x} >0$. So for any such triple all the variables are positive and we may write:
\[ \varphi(\qb,\rb) + \delta(\qb,\qprb) - \varphi(\qb,\rprb) 
= \sum_{x,y,j} p_{x,y} \, q_{j|x} \log \frac{q_{j|x}r'_{j|y}}{q'_{j|x}r_{j|y}} , \]
which, using that $\log t \geq 1-1/t$, can be bounded from below as follows: 
\[
\sum_{x,y,j} p_{x,y} \, q_{j|x} \left( 1- \frac{q'_{j|x}r_{j|y}}{q_{j|x}r'_{j|y}} \right) = 
1 - \sum_{y} p^{y} \sum_j \frac{r_{j|y}}{r'_{j|y}} \underbrace{ \sum_x p_{x|y}q'_{j|x} }_{=R_{j|y}(\qprb)= r'_{j|y}} 
= 1 - \sum_{y} p^y \sum_j r_{j|y} = 0 ,
\]
and the proof is complete.
\end{proof}

Now we are ready to prove that the alternating optimization process converges to the minimum.
\begin{proof}[Proof of Theorem \ref{thm:convergence}]
Consider the sequences $\qb^{(n)}$ and $\rb^{(n)}$ of alternating optimization started from some $\qb^{(0)} \in \ints(\Kq)$. Fix any pair $\qb \in \Kq$, $\rb \in \Kr$, and let $n$ be a positive integer. Using Proposition \ref{prop:3pt} for the triple $\qb, \rb^{(n-1)} \Qto \qb^{(n)}$ and Proposition \ref{prop:4pt} for the quadruple $\qb, \rb, \qb^{(n)} \Rto \rb^{(n)}$ we get that 
\begin{align*}
\delta\big( \qb, \qb^{(n)} \big) + \varphi\big( \qb^{(n)}, \rb^{(n-1)} \big) & 
\stackrel{\tiny{3-pt}}{=} \varphi\big( \qb,\rb^{(n-1)} \big);\\
\varphi\big( \qb,\rb^{(n)} \big) &
\stackrel{\tiny{4-pt}}{\leq}  \varphi(\qb,\rb) + \delta\big( \qb,\qb^{(n)} \big) .
\end{align*}
Since $\qb^{(n)} \in \ints(\Kq)$ holds for all $n$, we have $\delta\big( \qb,\qb^{(n)} \big)<\infty$. Therefore adding the two inequalities above results in   
\[ \varphi\big( \qb^{(n)}, \rb^{(n-1)} \big) + \varphi\big( \qb,\rb^{(n)} \big) 
\leq \varphi(\qb,\rb) + \varphi\big( \qb,\rb^{(n-1)} \big) .\]
Since 
$\varphi\big( \qb^{(n)}, \rb^{(n)} \big) \leq \varphi\big( \qb^{(n)}, \rb^{(n-1)} \big)$ 
by Proposition \ref{prop:R_ineq}, it follows that 
\begin{equation} \label{eq:5pt}
\varphi\big( \qb^{(n)}, \rb^{(n)}\big) 
+ \varphi\big( \qb, \rb^{(n)}\big) 
\leq \varphi\big( \qb, \rb \big)
+ \varphi\big( \qb, \rb^{(n-1)}\big) .
\end{equation} 
This is what Csisz\'ar and Tusn\'ady refer to as the \emph{5-point property} for the points $\qb, \rb, \rb^{(n-1)} \Qto \qb^{(n)} \Rto \rb^{(n)}$.

Note that the second term on either side is an element of the sequence $\varphi\big( \qb, \rb^{(n)}\big) \geq 0$. First we assume that these elements are all finite. Then for any $\eps>0$ there must be infinitely many $n$ such that
\[ \varphi\big( \qb, \rb^{(n)}\big) \geq \varphi\big( \qb, \rb^{(n-1)}\big) - \eps , \]
otherwise the sequence would converge to $-\infty$, contradicting that each element is nonnegative. For any such $n$ we get from \eqref{eq:5pt} that 
\[ \varphi\big( \qb^{(n)}, \rb^{(n)}\big) \leq \varphi\big( \qb, \rb \big) + \eps .\]
Since $\varphi\big( \qb^{(n)}, \rb^{(n)}\big)$ is monotone decreasing, it has a limit that must satisfy 
\[ \lim_{n \to \infty} \varphi\big( \qb^{(n)}, \rb^{(n)}\big) \leq \varphi\big( \qb, \rb \big) + \eps \]
for any positive $\eps$, and hence for $\eps=0$ as well.

If, on the other hand, $\varphi\big( \qb, \rb^{(n)}\big) = \infty$ for some $n$, then $\varphi(\qb,\rb)=\infty$ follows from the 4-point property as $\delta\big( \qb,\qb^{(n)} \big) < \infty$, and we have the same conclusion: the limit is at most $\varphi(\qb,\rb)$.

Since this holds for any $\qb$ and $\rb$, it follows that the limit must be the minimum of $\varphi$. 
\end{proof}

\section{Convex corners} \label{sec:corner}

Convex corners are downward closed, convex subsets of $[0,\infty)^n$. It is possible to define entropy functions for convex corners, and this general theory was known to include the notion of graph entropy (via the vertex packing polytope, a convex corner associated to a graph). In this section we will show that conditional graph entropy can also be expressed as the entropy of an associated convex corner. Moreover, we will even define a dual problem in the form of another convex corner. 

Besides revealing a nice theoretical connection to a general theory, this also has significant practical implications: the dual problem provides a way to check optimality in the primal problem, even yielding an error bound. The error bound comes in particularly handy when combined with alternating optimization: we can stop at any time through the iterations $\rb^{(n)}$ and compute this error bound $\delta$, which then ensures that we are at most $\delta$ away from the optimum: 
\[  \phir\big( \rb^{(n)} \big) -\delta \leq H_G(X|Y) \leq \phir\big( \rb^{(n)} \big) .\]

We start by recalling the basic concepts regarding convex corners.

\subsection{Entropy of convex corners}
\begin{definition} \label{def:convex_corner}
A set $K \subset [0,\infty)^\Xc$ is said to be \emph{downward closed} if the following property\footnote{Here $0 \leq \aprb \leq \ab$ means that $0 \leq a'_x \leq a_x$ for each $x$. As before, we use the notation $\ab = \big( a_x \big)_{x \in \Xc}$ for points in $\Rb^{\Xc}$.} holds: 
\begin{equation*} 
\mbox{if } \ab \in K \mbox{, then } \aprb \in K 
\mbox{ for all points } 0 \leq \aprb \leq \ab .
\end{equation*}
Similarly, $K$ is \emph{upward closed} if $\aprb \in K$ whenever $\ab \in K$ and $\aprb \geq \ab$.

We say that $K \subset [0,\infty)^\Xc$ is a \emph{convex corner} if $K$ is compact, convex, and downward closed. Usually $K$ is also required to have nonempty interior, or equivalently, to contain a point with strictly positive coordinates. 

Given a random variable $X$ and the corresponding probabilities $p_x$, $x \in \Xc$, let $\phia$ denote the following $[0,\infty)^\Xc \to [0,\infty]$ function: 
\[ \phia \colon \ab \mapsto -\sum_x p_x \log a_x .\]
Note that $\phia$ depends on the distribution of $X$, and we write $\phiaX$ when we want to emphasize this dependence. The entropy is defined as the minimum of $\phiaX$ over $K$:
\[ H_K(X) \defeq \min_{\ab \in K} \phiaX(\ab) .\]
The function $H_K(\cdot)$, defined for random variables on $\Xc$, is sometimes referred to as the \emph{entropy function} corresponding to the convex corner $K$. It can be seen that the entropy function $H_K(\cdot)$ 
uniquely determines $K$. 
\end{definition}

A related useful concept is the \emph{antiblocker} $K^\ast$ of a convex corner $K$:
\[ K^\ast \defeq \big\{ \bb \geq 0 \, : \, \sum_x a_x b_x \leq 1 \mbox{ for all } \ab \in K \big\} .\] 
One can show that $K^\ast$ is also a convex corner, $\big( K^\ast \big)^\ast = K$, and $H_K(X)+H_{K^\ast}(X)=H(X)$.
For these and further properties of $H_K(\cdot)$, see \cite[Sections 4.1 \& 6]{survey2} and \cite[Section 5]{Vrana2021}.

We will also use the following notations: for $\ab, \bb \in [0,\infty)^{\Xc}$ let $\ab\bb$ denote the vector $\big( a_x b_x \big)_{x \in \Xc}$ (coordinate-wise multiplication). Similarly, $\ab^{-1}$ denotes the vector with coordinates $1/a_x$ (provided that each $a_x$ is positive). Furthermore,
\[ \ab K \defeq \big\{ \ab \bb \, : \, \bb \in K \big\} \mbox{ and } 
K^{-1} \defeq \big\{ \bb^{-1} \, : \, \bb \in K \big\} .\]
Note that $\phia(\ab \bb) = \phia(\ab) + \phia(\bb)$ and $\phia(\ab^{-1}) = - \phia(\ab)$. Finally, we denote the vector $\big( p_x \big)_{x \in \Xc}$ by $\pb$. Then $\phia(\pb)=-\sum_x p_x \log p_x$ is the entropy of $X$.

\subsection{Primal problem}
Now we introduce the \emph{a-problem} $\min_{\Ka} \phia$, which is, in fact, an equivalent formulation of the r-problem. We have already defined the function to minimize: $\phia$. Next we define the convex corner $\Ka$ (associated to $X,Y,\Jc$) simply as the smallest downward closed set containing $A(\Kr)$. 
\begin{definition}
Let
\[ \Ka \defeq \big\{ \ab \in \Rb^\Xc \, : \, 
0 \leq \ab \leq A( \rb ) \mbox{ for some } \rb \in \Kr \big\} .\]
%
\end{definition}
\begin{proposition} \label{prop:corner_entropy}
The set $\Ka$ is a convex corner and $\ds H_{\Ka}(X)=\min_{\Ka} \phia=H_G(X|Y)$. 
\end{proposition}
\begin{proof}
The key observation is that $A_x$ is a concave function for each $x$, which follows immediately from the following claim:
%
let $\al_1,\ldots,\al_k \geq 0$ with $\al_1+\cdots+\al_k \leq 1$; then 
\[ f(t_1,\ldots,t_k) \defeq t_1^{\al_1} \cdots t_k^{\al_k} \]
is a concave function in the positive orthant 
$\big\{ (t_1,\ldots,t_k) \, : \, t_1,\ldots,t_k \geq 0 \big\}$.
%
Indeed, it is easy to see that the Hessian of $f$ is given by 
\[ 
H_{i,j}= 
\begin{cases}
\frac{\al_i\al_j}{t_i t_j} f(\tb) & \mbox{ if } i \neq j;\\
\frac{\al_i(\al_i-1)}{t_i^2} f(\tb) & \mbox{ if } i=j .
\end{cases} 
\]
Then for a vector $\ub=(u_i)$ we have 
\[ \ub H \ub^\top \big/ f(\tb) 
= \bigg( \sum_i \frac{\al_i u_i}{t_i} \bigg)^2 
- \sum_i \frac{\al_i u_i^2}{t_i^2} \leq 0 \]
by the Cauchy--Schwarz inequality, proving that the Hessian is negative semidefinite.
%

Since each $A_x$ is the sum of such functions, it is concave as well.

Now suppose that $\ab, \aprb \in \Ka$. By definition, there exist $\rb, \rprb \in \Kr$ such that $\ab \leq A(\rb)$ and $\aprb \leq A(\rprb)$. Then for any $t \in (0,1)$ and for any $x$ we have 
\[ t a_x + (1-t) a'_x \leq t A_x(\rb) + (1-t) A_x(\rprb) \leq A_x\big( t \rb + (1-t)\rprb \big) ,\]
where the second inequality is due to the concavity of $A_x$. It follows that the convex combination 
\[ t \ab + (1-t) \aprb \leq A\big( \underbrace{t \rb + (1-t)\rprb}_{\in \Kr} \big) \]
also lies in $\Ka$, proving the convexity of $\Ka$.

Since $\Kr$ is compact and $A$ is continuous, the image $A(\Kr)$ is also compact, and hence so is $\Ka$. Furthermore, if $\rb \in \Krst \supseteq \ints(\Kr) \neq \emptyset$, then $A_x(\rb)>0$ for each $x$, so $\Ka$ has a nonempty interior.

Finally, to see that $H_{\Ka}(X)=H_G(X|Y)$, it suffices to show that 
\[ \min_{\Ka} \phia = \min_{\Kr} \phir ,\]
which follows immediately from $\phir = \phia \circ A$ and the monotonicity of $\phia$: if $\ab \leq A(\rb)$, then 
\[ \phia(\ab) \geq \phia\big( A(\rb) \big) = \phir(\rb) \]
with equality when $\ab = A(\rb) \in \Ka$. 
\end{proof}
\begin{remark}
We make some comments regarding the a-problem.
\begin{itemize}
\item Note that $\Ka$ depends only on $\Jc$ (or the graph) and the conditional distributions of $Y \, | \, X=x$ (but not on the distribution of $X$). In the unconditioned case $\Ka$ is the vertex packing polytope $\VP(G)$ of the graph: the convex hull of the indicator functions of the independent sets. In general, $\Ka$ is not necessarily a polytope, it may be a more complicated convex set with ``curvy'' boundary. For an example, see Figure \ref{fig:ex_or} in Section \ref{sec:ex}. 
\item It is easy to see that $\phia$ is a strictly convex function over $(0,1]^\Xc$. Consequently, the a-problem always has a unique minimum point.
\item Note that the dimension of the a-problem is usually much smaller than that of the q-problem or the r-problem. However, the domain is not a polytope in this case and the complexity of the a-problem is, in some sense, hidden in the definition of the domain.
\end{itemize}
\end{remark}

\subsection{Dual problem} \label{sec:dual}
Now we introduce another convex corner that will lead to a dual problem. To this end, for each $j$ we define a function $\tau_j \colon [0,\infty)^\Xc \times [0,\infty)^\Yc \to \Rb$: for $\bb \in [0,\infty)^\Xc$ and $\tb \in [0,\infty)^\Yc$ we set  
\[ \tau_j(\bb,\tb) \defeq \sum_{x \in j} p_x b_x \prod_y \big( t_y \big)^{p^{y|x}} .\]
Note that $\sum_y p^{y|x} = 1$, hence $\tau_j(\bb,\tb)$ is homogeneous in $\tb$: for any scalar $\la>0$ we have $\tau_j(\bb,\la \tb) = \la \tau_j(\bb,\tb)$. 
\begin{definition} \label{def:L}
Let 
\begin{equation} \label{eq:Lj_def}
L_j \defeq \left\{ \bb \in [0,\infty)^\Xc \, : \, \forall \tb \, \tau_j(\bb,\tb) \leq \sum_y p^y t_y \right\} 
= \left\{ \bb \in [0,\infty)^\Xc \, : \, \tau_j(\bb,\tb) \leq 1 \mbox{ for all } \tb \mbox{ with } \sum_y p^y t_y=1 \right\} .
\end{equation}
%
Finally, we define $\La$ as the intersection of all $L_j$:
\[ \La \defeq \bigcap_j L_j .\] 
\end{definition}
To see that $L$ is a convex corner, notice that for any given $j$ and $\tb$, the points $\bb$ for which $\tau_j(\bb,\tb) \leq 1$ form a downward closed polyhedron, and $L$ is the intersection of such sets.
\begin{remark}
For graph entropy (i.e., the unconditioned case $|\Yc|=1$) it can be seen easily that 
\[ L = \big\{ \bb \, : \, \forall j \, \sum_{x \in j} p_x b_x \leq 1 \big\} = (\pb \Ka)^\ast .\]
\end{remark}
As the following lemma shows, the containment $\La \subseteq (\pb \Ka)^\ast$ is true in general.
\begin{lemma} 
For any $\ab \in \Ka$ and $\bb \in \La$ we have $\ds \sum_x p_x a_x b_x \leq 1$. In other words, $\La \subseteq (\pb \Ka)^\ast$. 
\end{lemma}
\begin{proof}
We have $\ab \leq A(\rb)$ for some $\rb \in \Kr$, thus 
\[ \sum_x p_x a_x b_x \leq \sum_j \sum_{x \in j} p_x b_x  \, \prod_y \big( r_{j|y} \big)^{p^{y|x}} 
\leq \sum_j \sum_y p^y r_{j|y} = \sum_y p^y \underbrace{\sum_j r_{j|y}}_{=1} = \sum_y p^y = 1, 
\]
where we used that $\bb \in L_j$ for any given $j$, and hence $\tau_j(\bb,\tb) = \sum_{x \in j} p_x b_x \prod_y \big( t_y \big)^{p^{y|x}} \leq \sum_y p^y t_y$ holds for $t_y=r_{j|y}$.
\end{proof}
\begin{corollary} \label{cor:half_duality}
For any $\ab \in \Ka$ and $\bb \in \La$ we have 
\[ \phia(\ab) + \phia(\bb) \geq 0 .\]
In other words, 
\begin{equation} \label{eq:minKL}
\min_{\Ka} \phia + \min_{\La} \phia \geq 0 .    
\end{equation}
\end{corollary}
\begin{proof}
Since $-\log$ is convex and monotone decreasing, by the above lemma we have 
\[ \phia(\ab) + \phia(\bb) = \phia(\ab \bb) = -\sum_x p_x \log(a_x b_x) \geq -\log\left( \sum_x p_x a_x b_x \right) \geq -\log(1) = 0 .\]
\end{proof}

We will shortly see that \eqref{eq:minKL} actually holds with equality. In order to prove this, let us consider the set $\La^{-1} = \big\{ \bb^{-1} \, : \, \bb \in \La \big\}$. Since $\La$ is convex and downward closed, it follows easily that $\La^{-1}$ is convex and upward closed (using the convexity of $t \mapsto 1/t$ for $t>0$). The key observation is that $\Ka$ and $\La^{-1}$ always have a common point.
\begin{theorem} \label{thm:duality}
The intersection of the downward closed convex set $\Ka$ and the upward closed convex set $\La^{-1}$ is a single point $\ab$, where $\phia$ takes its minimum over $\Ka$ and its maximum over $\La^{-1}$. Then 
\[ H_G(X|Y) = \phia(\ab) = \min_{\Ka} \phia = \max_{\La^{-1}} \phia = - \min_{\La} \phia .\]
Furthermore, $\Ka$ and $\La^{-1}$ are separated by a hyperplane with normal vector $\pb \ab^{-1}=\big( p_x/a_x \big)_{x \in \Xc}$. 
\end{theorem}
Before we present the proof, recall that the mappings $Q$ and $R$ ``jump'' between the q-problem and r-problem in a way that the function value decreases. In what follows we will focus on the r-problem and the corresponding stepping map 
\[ \Fr \defeq R \circ Q \colon \Kr \to \Kr .\]
\begin{proposition} \label{prop:min_fixed}
Every minimum point of $\phir$ must be a fixed point of $\Fr$.
\end{proposition}
\begin{proof}
Combining Propositions \ref{prop:R_ineq} and \ref{prop:Q_ineq} gives that for any $\rb \in \Krst$ we have 
\[ \phir\big( \Fr(\rb) \big) \leq \phir(\rb) \]
with equality if and only if $\rb$ is a fixed point of $\Fr$. In other words, if $\rb$ is not a fixed point, then we have strict inequality and hence $\phir(\rb)$ cannot be the minimum.
\end{proof}
\begin{proof}[Proof of Theorem \ref{thm:duality}]
Since $\phia(\bb^{-1}) = -\phia(\bb)$, we have
\[ \max_{\La^{-1}} \phia = - \min_{\La} \phia \stackrel{\eqref{eq:minKL}}{\leq} \underbrace{\min_{\Ka} \phia}_{=H_G(X|Y)} = \phia(\ab) ,\]
where $\ab$ denotes the unique minimum point\footnote{Since $\Ka$ is compact and $\phia \colon \Ka \to [0,\infty]$ is continuous, its minimum is attained at some $\ab \in \Ka$. The minimum is finite, so we have $a_x>0$ for each $x$. In that region $\phia$ is strictly convex, therefore $\ab$ is indeed unique.} of $\phia$ over $\Ka$. It remains to be shown that $\ab \in \La^{-1}$, implying the only missing inequality $\max_{\La^{-1}} \phia \geq \phia(\ab)$ and confirming that $\Ka \cap \La^{-1} = \{ \ab \}$. (Note that the $\phia(\aprb) > \phia(\ab)$ for any $\aprb \in \Ka \sm \{\ab\}$, and hence $\aprb \notin \La^{-1}$.) Also, the gradient of $\phia$ at $\ab$ is $-\pb \ab^{-1}$, so the hyperplane through $\ab$ that separates the convex sets $\Ka$ and $\La^{-1}$ must be the one with normal vector $\pb \ab^{-1}$.

In order to prove that $\ab \in \La^{-1}$, let $\rb$ be such that $\ab = A(\rb)$ so that $\phir(\rb)=\phia(\ab)$, that is, $\rb$ minimizes $\phir$ over $\Kr$. (Note that $\rb$ may not be unique.) By Proposition \ref{prop:min_fixed}, $\rb$ is a fixed point of $\Fr$. That is, for $\qb \defeq Q(\rb)$ we have $R(\qb)=\rb$.

For brevity, we write $\partial_{j|y}$ for the partial derivative w.r.t.\ the variable $r_{j|y}$, and $g_{j,x}$ for the product in the definition of $A_x$, that is: 
\[ g_{j,x} \defeq \prod_y \big( r_{j|y} \big)^{p^{y|x}} \mbox{ so that } 
A_x(\rb) = \sum_{j \ni x} g_{j,x} .\] 
If $r_{j|y}>0$, then we have 
\[\partial_{j|y} A_x \big( \rb \big) =  
\begin{cases}
0 & \mbox{ if } x \notin j ;\\
\frac{p^{y|x}}{r_{j|y}} g_{j,x}
& \mbox{ if } x \in j ;
\end{cases} 
\]
and hence
\[
\partial_{j|y} \phir \big( \rb \big) =  
 - \sum_x p_x \frac{ \partial_{j|y} A_x(\rb)}{A_x(\rb)} 
=  - \sum_{x \in j} 
\frac{p_x p^{y|x}g_{j,x}}{r_{j|y} A_x(\rb)} 
= - \frac{ p^y \sum_{x \in j} p_{x|y} q_{j|x}}{r_{j|y}} 
= -\frac{p^y R_{j|y}(\qb)}{r_{j|y}} = - p^y ,
\]
where we used that $R(\qb)=\rb$. 

Now we fix a $\jh \in \Jc$ and a vector $\tb=(t_y)$ with $t_y \geq 0$. Then we perturb $\rb$ in the coordinates $\jh|y$ as follows: for a given $\eps>0$ we define the perturbed vector $\rb^\eps$ as  
\[ r^\eps_{j|y} \defeq 
\begin{cases}
r_{j|y} + \eps t_y & \mbox{if } j=\jh ;\\
r_{j|y}            & \mbox{if } j \neq \jh .
\end{cases} \]
Then $A_x$ does not change for $x \notin \jh$. As for $x \in \jh$, we claim that 
\begin{equation} \label{eq:change}
A_x(\rb^\eps) - A_x(\rb) \geq  \eps \prod_y \big( t_y \big)^{p^{y|x}} .
\end{equation}
%
To see this, notice that the function $f \colon \ub \mapsto \prod_y \big( u_y \big)^{p^{y|x}}$ is 
concave (as we have shown it in the proof of Proposition \ref{prop:corner_entropy}) and homogeneous (of degree $1$) to conclude that 
\[ f(\ub+\eps \tb)= 2 f \left( \frac{\ub+\eps \tb}{2} \right) \geq 2 \frac{ f(\ub) + f(\eps \tb) }{2} = f(\ub) + \eps f(\tb) ,\] 
which implies \eqref{eq:change}. It follows that 
\begin{equation*} 
\lim_{\eps \to 0+} \, \frac{\phir(\rb^\eps)-\phir(\rb)}{\eps}
= \lim_{\eps \to 0+} \, \frac{\phia( A(\rb^\eps) )-\phia( A(\rb)) }{\eps}
\leq  -\sum_{x \in \jh} \frac{p_x}{a_x} \prod_y
\big( t_y \big)^{p^{y|x}} 
= - \tau_{\jh} \big( \ab^{-1}, \tb \big) .
\end{equation*}
Note that $\rb^\eps \notin \Kr$ anymore because moving in direction $\tb$ violates the linear constraints of $\Kr$. So for each $y$ we decrease other (positive) coordinates $r_{j|y}$ by a total of $\eps t_y$ to get a point $\rhatb^\eps$ in $\Kr$. Since the partial derivative $\partial_{j|y} \phir$ is $-p^y$ for such positive coordinates, it is easy to see that we get the following:
\begin{equation} \label{eq:directional}
\lim_{\eps \to 0+} \, \frac{\phir(\rhatb^\eps)-\phir(\rb)}{\eps}
= - \tau_{\jh} \big( \ab^{-1}, \tb \big) + \sum_y p^y t_y.
\end{equation}
Since $\rhatb^\eps \in \Kr$, we have $\phir(\rhatb^\eps) \geq \min_{\Kr} \phir = \phir(\rb)$ for each $\eps$. So the above limit must be nonnegative, that is, 
\[ \tau_{\jh} \big( \ab^{-1}, \tb \big) \leq \sum_y p^y t_y .\]
This holds for any $\tb \geq 0$, meaning that $\ab^{-1} \in L_{\jh}$. This can be done for any $\jh \in \Jc$, implying $\ab^{-1} \in \La$. 
\end{proof}

\subsection{The Orlitsky--Roche example} \label{sec:ex}

Orlitsky and Roche considered the following simple example, see \cite[Examples 2\&5]{orlitsky_roche}. Let $\Xc=\Yc=\{1,2,3\}$ with the distribution 
\[ p_{x,y}= 
\begin{cases}
1/6 & \mbox{ if } x \neq y; \\
0   & \mbox{ if } x = y.
\end{cases}
\]
Furthermore, let $G$ be the graph on the vertex set $\Xc$ containing a single edge $(1,3)$ so that $G$ has two maximal independent sets: $\{1,2\}$ and $\{2,3\}$. They showed that 
\begin{equation} \label{eq:or_entr}    
H_G(X|Y) = -\frac{2}{3} \left( \frac{1}{4}\log \frac{1}{4} + \frac{3}{4}\log \frac{3}{4} \right) \approx 0.37489 . 
\end{equation}
We will use this example to illustrate our results. We have
\[ 
p_x=p^y=1/3 \, (\forall x,y) ; \quad p_{x|y}=p^{y|x}=
\begin{cases}
1/2 & \mbox{ if } x \neq y; \\
0   & \mbox{ if } x = y.
\end{cases} \]
We will use the notations $\al \defeq \{1,2\}$ and $\be \defeq \{2,3\}$ for the independent sets so that $\Jc=\{\al,\be\}$. It means that the r-problem has six non-negative variables with the following constraints:
\[ r_{\al|1}+r_{\be|1}=1 ; \ 
r_{\al|2}+r_{\be|2}=1 ; \ 
r_{\al|3}+r_{\be|3}=1 .\]
Then the mapping $A$ is described by the following coordinate functions:
\begin{align*}
A_1(\rb) &= \sqrt{r_{\al|2} r_{\al|3}} \\
A_2(\rb) &= \sqrt{r_{\al|1} r_{\al|3}} + \sqrt{r_{\be|1} r_{\be|3}}\\
A_3(\rb) &= \sqrt{r_{\be|1} r_{\be|2}} 
\end{align*}
Next we describe the convex corner $\Ka$ associated to this example. Note that $A_1^2(\rb)+A_3^2(\rb) \leq r_{\al|2} + r_{\be|2} = 1$, and $A_2(\rb) \leq \sqrt{r_{\al|1}+r_{\be|1}} \cdot \sqrt{r_{\al|3}+r_{\be|3}} =1$ by Cauchy--Schwarz. It follows that for any $\ab \in \Ka$ we have $a_1^2 + a_3^2 \leq 1$ and $a_2 \leq 1$. Now fix $a_1, a_3$ such that $a_1^2 + a_3^2 \leq 1$, and let us try to find the largest possible corresponding $a_2$ value: 
\[ \mbox{let } w \defeq r_{\al|2} ; \mbox{then } 
r_{\be|2} = 1-w; 
r_{\al|3} = \frac{a_1^2}{w} ; 
r_{\be|1} = \frac{a_3^2}{1-w} .\]
Therefore 
\[ a_2 = A_2(\rb) = \sqrt{\frac{a_1^2}{w} \left(1- \frac{a_3^2}{1-w} \right) } + 
\sqrt{\frac{a_3^2}{1-w} \left(1- \frac{a_1^2}{w} \right) } .\]
We need to maximize this formula in the one free variable $w$. It is easy to see that when $a_1+a_3 \leq 1$, the maximum is always $1$, meaning that the boundary of $\Ka$ includes a triangle whose vertices are $(0,1,0)$; $(1,1,0)$ and $(0,1,1)$. When $a_1+a_3 > 1$ and $a_1^2+a_3^2 \leq 1$, we did not find a closed formula, but one can easily plot the maximum as a function of the parameters $a_1$ and $a_3$; see Figure \ref{fig:ex_or}. Note that when $a_1=a_3$, the maximum is always taken at $w=1/2$, so we get $a_2=2\sqrt{2}a_1\sqrt{1-2a_1^2}$ for the boundary of $\Ka$ in this cross section. 

\begin{figure}[ht]
\centering
\includegraphics[height=4in]{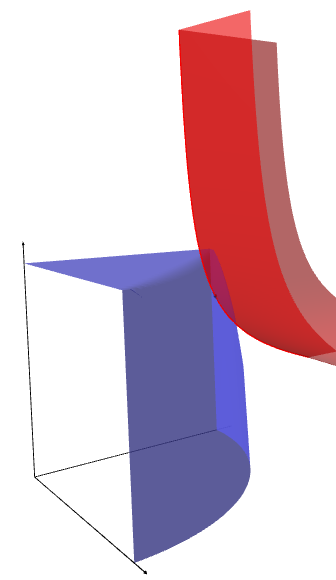} \hspace{0.5in} 
\includegraphics[height=4in]{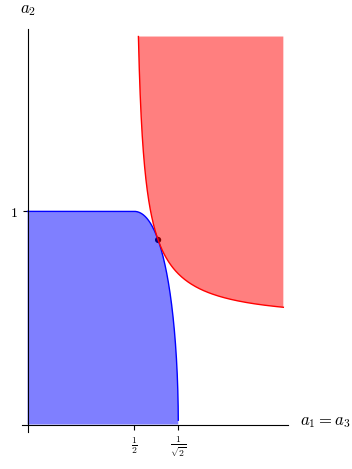} 
\caption{On the left: plots of the boundaries of $\Ka$ (blue) and $\La^{-1}$ (red). On the right: the two-dimensional cross section corresponding to the plane $a_1=a_3$. The black dot marks the unique intersection point, where $\phia$ takes its minimum over $\Ka$ and its maximum over $\La^{-1}$.}
\label{fig:ex_or}
\end{figure}

As for the convex corner $\La$ corresponding to the dual problem, first we need to work out the formulas for $\tau_\al$ and $\tau_\be$:
\begin{align*}
\tau_\al(\bb,\tb) &= \frac{1}{3} \left( b_1 \sqrt{t_2 t_3} + b_2 \sqrt{t_1 t_3} \right) ;\\
\tau_\be(\bb,\tb) &= \frac{1}{3} \left( b_2 \sqrt{t_1 t_3} + b_3 \sqrt{t_1 t_2} \right) .
\end{align*}
By Cauchy--Schwarz we have 
\[ \tau_\al(\bb,\tb) \leq \frac{1}{3} \sqrt{b_1^2+b_2^2} \sqrt{(t_1+t_2)t_3} \leq \frac{1}{3} \sqrt{b_1^2+b_2^2} \, \frac{t_1+t_2+t_3}{2} ,\]
which shows by \eqref{eq:Lj_def} that $\bb \in \La_\al$ provided that $b_1^2+b_2^2 \leq 4$. It is also easy to see that $\bb \notin \La_\al$ if $b_1^2+b_2^2 > 4$. Similar calculations show that $\bb \in \La_\be$ if and only if $b_2^2+b_3^2 \leq 4$. We conclude that 
\[ \La = \big\{ \bb=(b_1,b_2,b_3) \, : \, b_1^2+b_2^2 \leq 4 \mbox{ and } b_2^2+b_3^2 \leq 4 \big\} .\]
In Figure \ref{fig:ex_or} we plotted (the boundary of) $\La^{-1}$ instead of $\La$ to illustrate the fact that $\Ka$ and $\La^{-1}$ intersect in a single point (marked by a black dot in the figure). This intersection point is where $\phia$ takes its minimum over $\Ka$. The minimum points of the various problems are as follows.

The minimum of $\phiq$ is attained at the following point $\qb$:
\begin{align*}
q_{\al|1} &= 1 ; \, q_{\be|1} = 0 ; \\
q_{\al|2} &= q_{\be|2} = 1/2 ; \\
q_{\al|3} &= 0 ; \, q_{\be|3} = 1 .
\end{align*}
Then $\phir$ takes its minimum at the corresponding point $\rb=R(\qb)$:
\begin{align*}
r_{\al|1} &= 1/4 ; \, r_{\be|1} = 3/4 ; \\
r_{\al|2} &= 1/2 ; \, r_{\be|2} = 1/2 ; \\
r_{\al|3} &= 3/4 ; \, r_{\be|3} = 1/4 .
\end{align*}
Lastly, $\phia$ takes its minimum at 
\[ \ab=A(\rb)= \left( \sqrt{3/8}, \sqrt{3/4}, \sqrt{3/8} \right) .\] 
Using our results, one can easily verify that this is the optimal point in $\Ka$ by checking that $\ab^{-1} \in \La$, which indeed holds as 
\[ \mbox{for } \bb=\ab^{-1}=\left( \sqrt{8/3}, \sqrt{4/3}, \sqrt{8/3} \right) 
\mbox{ we have } b_1^2+b_2^2=b_2^2+b_3^2= 4/3 + 8/3=4 .\]
This confirms the value of $H_G(X|Y)$; see \eqref{eq:or_entr}. 

Finally, the table below shows the values $\phir\big( \rb^{(n)} \big)$ of the iterative process started from a random point $\qb^{(0)}$. We also included the error (i.e., the distance from the minimum) and our error bound based on the dual problem (see Theorem \ref{thm:error_bound} in Section \ref{sec:algorithm}). 
%
%
\begin{center}
\begin{tabular}{r|l|l|l}
$n$ & value & error & error bound (see Thm \ref{thm:error_bound})\\
    & $\phir\big( \rb^{(n)} \big)$ & $\phir\big( \rb^{(n)} \big) - H_G(X|Y)$ & $\delta\big( A\big( \rb^{(n)} \big)\big) $ \\
\hline
$5$  & $0.3749085763210158$ & $1.8\cdot 10^{-5}$  & $2.5\cdot 10^{-3}$\\
$10$ & $0.3748904169016328$ & $3.2\cdot 10^{-7}$  & $3.2\cdot 10^{-4}$\\
$15$ & $0.3748901019703158$ & $5.5\cdot 10^{-9}$  & $4.3\cdot 10^{-5}$\\
$20$ & $0.3748900965089192$ & $9.6\cdot 10^{-11}$ & $5.6\cdot 10^{-6}$\\
$25$ & $0.3748900964142102$ & $1.6\cdot 10^{-12}$ & $7.4\cdot 10^{-7}$\\
$30$ & $0.3748900964125679$ & $2.9\cdot 10^{-14}$ & $9.8\cdot 10^{-8}$\\
$35$ & $0.3748900964125393$ & $4.6\cdot 10^{-16}$ & $1.2\cdot 10^{-8}$\\
\hline
$H_G(X|Y)$ & $0.3748900964125389...$   
\end{tabular}
\end{center}
%

\subsection{Fractional chromatic number}

Given a convex corner $K$, it is natural to ask what the maximum of its entropy function is. That is, by varying the distribution of $X$, what is the maximal possible $H_K(X)$ we can get for a fixed $K$? In general one can say the following about this maximum entropy.
\begin{lemma}[see Corollary 1.2.21 in \cite{boreland_thesis}] \label{lem:tau}
Let $K \subset \Rb^\Xc$ be an arbitrary convex corner. Then 
\[ \max_X H_K(X) = \log \tau(K) ,\]
where $\tau(K)$ denotes the smallest $t \geq 1$ such that the constant $1/t$ vector lies in $K$. (Note that here $X$ can be any random variable on $\Xc$: its support may be a proper subset of $\Xc$.)
\end{lemma}
%

The question arises: is there a special meaning of $\tau(K)$ in our setting? In the unconditioned case, that is, for the vertex packing polytope $K=\VP(G)$, $\tau(K)$ is known to be equal to the fractional chromatic number of the graph \cite[Lemma 4]{survey2}. Is there a generalization of this result: does $\tau( \Ka )$ have a nice graph theoretic meaning in the conditional setting?
\begin{problem} \label{prob:max_entr}
Fix a graph $G$ equipped with a distribution on $\Yc$ at each vertex $x$ (described by $p_{y|x}$). Note that this determines the convex corner $\Ka$. Is it possible to give a (graph theoretic) description of $\tau( \Ka )$? This could lead to a notion generalizing the fractional chromatic number to \emph{measure-labelled graphs}.
\end{problem}
%

\section{Discussion of the algorithm} \label{sec:algorithm}

As we have seen in the introduction, one may start at any point $\qb^{(0)} \in \ints(\Kq)$ and alternate in applying the mappings $R$ and $Q$ to get a sequence \eqref{eq:alt_seq} with decreasing $\varphi$-values. In fact, Theorem \ref{thm:convergence} tells us that the values always converge to $\min \varphi(\qb,\rb) = H_G(X|Y)$. 

In this section, we provide an error bound for the algorithm, then propose a tweak for improving the running time, and finally analyze the rate of convergence in the unconditioned case of graph entropy.

\subsection{Error bound}
How long should we run the iterations? A natural stopping rule is to terminate the algorithm at a step where the drop in the $\varphi$-value gets below some threshold. Is there a way to know how far we are from the actual minimum? Using the dual problem defined in Section \ref{sec:dual}, we can easily get an error bound for any given $\rb^{(n)}$ we stop at. 
\begin{theorem} \label{thm:error_bound}
Let $\rb \in \Kr$ arbitrary and set $\ab=A(\rb)$. For each $j$ consider the following maximization problem: 
\begin{equation} \label{eq:delta_j}
1+\delta_j(\ab) \defeq \max_{\tb} \tau_j\big(\ab^{-1},\tb\big) 
= \max_{\tb} \sum_{x \in j} \frac{p_x}{a_x} \prod_y
\big( t_y \big)^{p^{y|x}} 
\mbox{ under the constraints }
t_y \geq 0; \, \sum_y p^y t_y= 1 .
\end{equation}
Then $\phia(\ab)=\phir(\rb)$ is at most $\delta(\ab) \defeq \max_j \delta_j(\ab)$ away from the minimum. More precisely,
\[ \phir(\rb) - H_G(X|Y) \leq \log\big( 1+\delta(\ab) \big) \leq \delta(\ab) .\]
In particular, $\ab$ (and hence $\rb$) is optimal if and only if $\delta(\ab) = 0$.
\end{theorem}
\begin{remark}
Note that each maximization is a convex optimization problem, whose dimension ($|\Yc|$) is small compared to that of the r-problem ($|\Yc| \cdot |\Jc|$) so we can solve them with high precision relatively fast.
\end{remark}
\begin{proof}
By definition, $\bb \defeq \big( 1+\delta(\ab) \big)^{-1}\ab^{-1}$ lies in $\La$. Therefore 
\[ H_G(X|Y) = -\min_\La \phia \leq -\phia(\bb) = \log\big( 1+\delta(\ab) \big) + \phia(\ab) .\]
\end{proof}

The table at the end of Section \ref{sec:ex} compares this error bound to the true error for the Orlitsky--Roche example.

\subsection{A tweak: deleting redundant sets}
The running time of the algorithm depends on two things: the time required to perform a single step and the number of steps required to get within the desired distance of the minimum. With one small tweak we can achieve significant gains for both at the same time.

First of all, note that at each step the algorithm performs $\Oc\big( |\Xc| \cdot |\Yc| \cdot |\Jc| \big)$ operations when computing $\rb^{(n)}=R\big( \qb^{(n)} \big)$ and $\qb^{(n+1)}=Q\big( \rb^{(n)} \big)$. 

In examples there are often a large number of (independent sets) $j$ that are actually not ``used'' at the optimal $\qb$ and $\rb$ in the sense that $q_{j|x}=0$ and $r_{j|y}=0$ for all $x,y$. For any such $j$, these variables will converge to $0$ through the iterations. To speed things up, we may want to detect such \emph{redundant} sets $j$ early and set the corresponding variables to $0$. Note that these variables remain to be $0$ from this point on, so we may remove such a $j$ from $\Jc$ and proceed with the iterations using a smaller set $\Jc$. This immediately reduces the computational complexity for each subsequent step. Moreover, it typically results in a better rate of convergence as well: without redundant sets, the error usually decays at a faster rate. Consequently, this version of the algorithm often requires considerably fewer steps to reach the desired precision. (This phenomenon will be illustrated for graph entropy both by an example and an analysis.) 

However, when the algorithm terminates and outputs an (approximate) minimum point for some subsystem $\Jc$ of the original $\Jcor$, we should justify that all deletions we made along the way were indeed necessary. So we take the corresponding point $\ab$ and perform our optimality check/error bound calculations: we compute $\delta_j(\ab)$ as in \eqref{eq:delta_j}. For each $j \in \Jc$ we should get a negative number or a very small positive number, confirming that we are indeed close to the minimum point of the problem corresponding to the subsytsem $\Jc$. If $\delta_j(\ab) \leq 0$ for all deleted sets $j \in \Jcor \setminus \Jc$, it means that we cannot do better even if we used the deleted sets. If, on the other hand, $\delta_j(\ab) > 0$ for some of the deleted sets $j$, then we should ``re-activate'' them (i.e., add them back to $\Jc$).

So we propose the following \textbf{tweaked version of the iterative process}.
\begin{mdframed}
\begin{itemize}
\item Set $r^{(0)}_{j|y} = 1 / |\Jc|$ for each $j$ and $y$. Note that $\rb^{(0)} \in \Kr$.
\item Set $\epsact = 10^{-3} |\Yc| / |\Jc|$.
\item At step $n$:
\begin{itemize}
    \item compute $\rb^{(n-1)} \Qto \qb^{(n)} \Rto \rb^{(n)}$;
    \item for any $j$ with $\sum_y r^{(n)}_{j|y} < \epsact$, remove $j$ from $\Jc$ and delete the corresponding variables $r^{(n)}_{j|y}$ for all $y$;
    \item for each $y$, re-normalize the remaining variables $r^{(n)}_{j|y}$, $j \in \Jc$ such that  
    \[ \mbox{the constraint } \sum_j r^{(n)}_{j|y} = 1 \mbox{ is satisfied again.} \]    
\end{itemize}
\item Compute the value $\phir(\rb^{(n)})$ after every $10$ steps, and terminate the iterations when this value, compared to the previous one, decreases by less than some small $\epsprec$ (say, $10^{-15}$).
\item Set $\ab = A\big( \rb^{(n)} \big)$.
\item Compute $\delta_j(\ab)$ as in \eqref{eq:delta_j} for all $j \in \Jc$ as well as for all previously deleted sets $j$.
\item If $\delta_j(\ab) \leq 0$ for each deleted $j$, then return $\phia(\ab)$ with the error bound $\ds \max_{j \in \Jc} \delta_j(\ab)$.
\item Otherwise, for each deleted $j \in \Jcor \setminus \Jc$ with $\delta_j(\ab) > 0$, add $j$ back to $\Jc$ and create the corresponding variables $r^{(n)}_{j|y}$ for each $y$, setting them to some small positive values.\footnote{Any values work but the following choice should guarantee that we get a smaller $\phir$-value right after restart: set $r^{(n)}_{j|y} \defeq \eps t_y$ for a sufficiently small $\eps$, where $\tb$ denotes the vector at which \eqref{eq:delta_j} takes its maximum.} Then re-normalize as before so that $\rb^{(n)} \in \Kr$ holds again. Finally, restart the iterations, this time with no set-deletions.
\end{itemize}
\end{mdframed}

Normally, we set $\epsact$ to be fairly small so that it is extremely unlikely that we unjustifiably delete a set $j$, and the check at the end should (essentially always) confirm this. 

Our implementation in Python is available on GitHub \cite{ge_github}.

\subsection{An example}

The next example shows how detecting redundant sets can speed the convergence up. 
\begin{example} \label{ex:dodeca}
Let $G$ be the dodecahedral graph: a $3$-regular graph with $20$ vertices and $30$ edges; see Figure \ref{fig:dodeca}. It has $295$ maximal independent sets. For a uniform $X$ we have 
\[ H_G(X) = \log \frac{5}{2}.\]
This can be seen easily using that $|j| \leq 8$ for each $j \in \Jc$ and that one can find five independent sets $j_1,\ldots,j_5$ such that each vertex is contained in exactly two of them (and hence $(2/5)\eb \in \Ka = \VP(G)$, where $\eb$ is the all-ones vector). 


%
\begin{figure}[ht]
\centering
\includegraphics[width=2in]{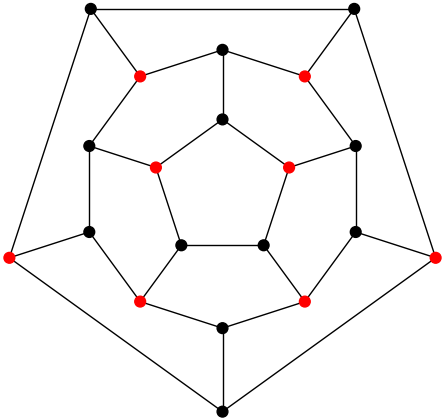}
\caption{The graph of the dodecahedron. The red vertices form an independent set of size $8$. By ``rotation'' one can get five independent sets in a way that each vertex is contained in exactly two of them.}
\label{fig:dodeca}
\end{figure}

Starting from a random point $\qb^{(0)}$, the blue dots below show the value $\phir\big( \rb^{(n)} \big)$ for each iteration $n=1,\ldots, 75$.
\begin{center}
\includegraphics[width=3in]{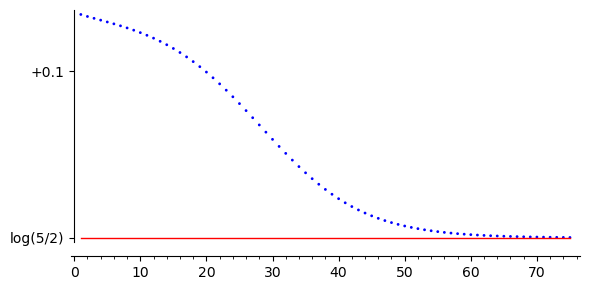}    
\end{center}
For comparison, we run the process from the same starting point, but this time deleting a set $j$ if $r^{(n)}_j$ gets below $\eps_{\mbox{\scriptsize{act}}} \defeq 2^{-20} \approx 10^{-6}$. Up to $n=75$ only $23$ sets were deleted and there was little difference in the value compared to the plot above. Afterwards the deletion rate accelerated and by step $n=101$ all but the five independent sets of size $8$ were deleted. The figure below compares the values in the two cases after step $80$. (We used blue dots for the original process with no set-deletions and green dots for the one with set-deletions.)
\begin{center}
\includegraphics[width=3in]{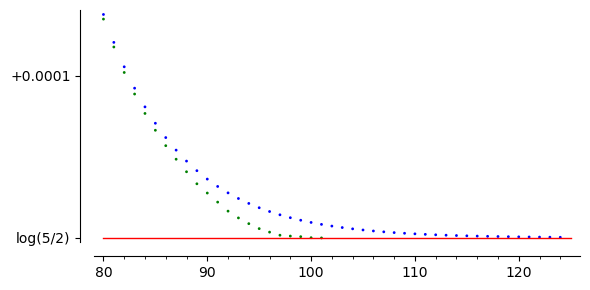}    
\end{center}
The original algorithm (blue dots) needed $239$ steps to get within distance $10^{-13}$ of the true minimum $\log(5/2)$, while the refined process (green dots) reached this threshold  after only $108$ iterations. We plotted the distance to the minimum in a logarithmic scale below: the horizontal axis shows the number of steps, while the vertical axis shows $-\log_{10}$ of the distance (i.e., the number of precise decimal digits essentially).
\begin{center}
\includegraphics[width=3in]{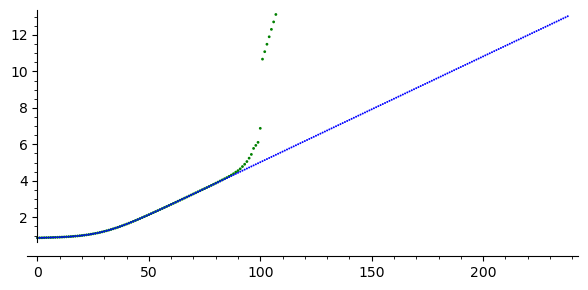}    
\end{center}
We see that there is a considerable leap in precision at the point when all $290$ ``redundant'' independent sets have been deleted. Both versions eventually settle into a phase where the precision (``number of precise digits'') grows at a linear rate. The tweaked version clearly exhibits a faster rate. In fact, the analysis in the next section will reveal that this faster rate is $2 \cdot \lg(8/5) \approx 0.408$ compared to the rate $\lg(8/7) \approx 0.058$ of the original version. 
\end{example}
%

\subsection{Rate of convergence for graph entropy} 

As we have seen in the example of the previous section, the precision of the iterative algorithm appears to grow at some linear rate (for steps $n \geq n_0$). The following analysis confirms this observation and explains how one can determine this (``eventual'') rate in the unconditioned setting (i.e., graph entropy). Rigorous proofs would make the analysis undesirably long and technical so in this section we settle for only sketching the arguments. 

\subsubsection*{Formulas for graph entropy}
In the special case of graph entropy (i.e., when $|\Yc|=1$ so there is only one $y$) the formulas simplify considerably. First of all, we have $p^y=p^{y|x}=1$ and $p_{x,y}=p_{x|y}=p_x$, and we may omit $y$ in the indices. So $\rb$ now denotes a point $(r_j)_{j \in \Jc}$ in the set 
\[ \Kr = \bigg\{ \rb= \big( r_j \big) \, : \, 
r_j \geq 0; \, \sum_{j} r_j = 1 \bigg\} \subset \Rb^{\Jc} .\]
Furthermore, we have the following simple formulas:
\begin{align*}
R_j(\qb) &= \sum_{x \in j} p_x \, q_{j|x} ;\\
A_x(\rb) &= \sum_{j \ni x} r_j ;\\
Q_{j|x}(\rb) &= 
\begin{cases}
0  & \mbox{ if } x \notin j ;\\
r_j \big/ A_x(\rb) & \mbox{ if } x \in j ;
\end{cases} \\
\phir(\rb) &= - \sum_x p_x \log A_x(\rb) 
=  - \sum_x p_x \log \sum_{j \ni x} r_j .
\end{align*}
So $A$ is simply a linear $\Rb^\Jc \to \Rb^\Xc$ map corresponding to the following matrix $M \in \Rb^{\Jc \times \Xc}$:
\[ M_{x,j} \defeq 
\begin{cases}
1 & \mbox{if } x \in j ;\\
0 & \mbox{if } x \notin j .
\end{cases} 
\]
That is, the columns of $M$ are the indicators functions of the sets $j$, and we have $A(\rb)=M \rb$.

As for the stepping map $\Fr \colon \rb \mapsto \rprb$ for the r-problem, we have 
\begin{equation} \label{eq:Fr_special}
r'_j = \bigg( 
\underbrace{\sum_{x \in j} \frac{p_x}{A_x(\rb)}}_{\Delta_j(\rb) \defeq }
\bigg) r_j .
\end{equation}
It follows that if $\rb$ is a fixed point of $\Fr$ (i.e., $r'_j=r_j$ for each $j$), if and only if $\Delta_j(\rb)=1$ for any $j$ with $r_j>0$.

It is worth mentioning that if $\partial_j$ denotes the partial derivative w.r.t.\ the variable $r_j$, then we have 
\[ \partial_j \phir(\rb) 
= - \sum_x p_x \frac{\partial_j A_x(\rb)}{A_x(\rb)} 
=  - \sum_{x \in j} \frac{p_x}{A_x(\rb)} = -\Delta_j(\rb) .
 \]
So what the stepping map $\Fr(\rb)$ does in this unconditioned setting is simply multiply $\rb$ (coordinate-wise) by the negative of the gradient $\nabla \phir(\rb)$.

\subsubsection*{Case of no redundant sets}
We start our analysis with the case when each $j$ is ``used'' ($r_j>0$) at the minimum point $\rb$ of $\phir$. This is always the case in the tweaked version of the algorithm which ensures that all redundant sets are eventually deleted. Note that $\Delta_j(\rb)=1$ for all $j$ in this case, and hence the gradient $\nabla \phir (\rb) = -\eb$ for the all-ones vector $\eb$.

For a vector $\vb$ we will use the notation $D(\vb)$ for the corresponding diagonal matrix. In particular, $D(\rb) \in \Rb^{\Jc \times \Jc}$ is the diagonal matrix with entries $r_j$, while $D(\pb \ab^{-2}) \in \Rb^{\Xc \times \Xc}$ is the diagonal matrix with entries $p_x/a_x^2$.
\begin{lemma} \label{lem:matrix_N}
Assume that $\rb$ is a minimum point of $\phir$ and that each $r_j>0$. Set $\ab \defeq A(\rb) = M \rb$ so that $\ab$ is the minimum point of $\phia$. Let 
\[ N \defeq D(\rb) M^\top D(\pb \ab^{-2}) M .\]
Then $N \in \Rb^{\Jc \times \Jc}$ is a square matrix with nonnegative entries and with the following properties: 
\begin{itemize}
\item in each column the sum of the entries is $1$ (and hence $1$ is an eigenvalue); 
\item $N$ is diagonalizable with eigenvalues in $[0,1]$;
\item $\ker N = \ker M$.
\end{itemize}
\end{lemma}
The proof of the lemma can be found at the end of the section. 

\begin{claim*}
The rate of convergence is governed by the smallest nonzero eigenvalue $\lam$ of $N$: 
\begin{equation} \label{eq:conv_rate}
\phir\big( \rb^{(n)} \big) = \min_{\Kr} \phir + O\big( (1-\lam)^{2n} \big) .
\end{equation}
So the rate of growth for the precision is $-2 \lg(1-\lam)$.    
\end{claim*}
\begin{example*}
The dodecahedral graph of Example \ref{ex:dodeca} has five independent sets of size $8$. Note that their pairwise intersections are of size $2$. Independent sets of smaller size are all redundant so let $\Jc$ be the set of these five sets. Then we have $r_j =1/5$ for all $j$ and $a_x=2/5$ for all $x$. It follows that each diagonal entry of $N$ is equal to $1/2$, while all other entries are equal to $1/8$. Therefore, the eigenvalues (with multiplicity) are  $1;\ 3/8;\ 3/8;\ 3/8;\ 3/8$. So $\lam=3/8$ and we get that the rate of growth for the precision is $-2 \lg(5/8) + o(1)$, which is consistent with our numerical findings presented earlier.
\end{example*}

Now we will sketch the proof of the claim. For the sake of simplicity we assume that $\ker M = \{0\}$. 
In this case $\rb$ is the unique minimum point of $\phir$ and hence $\rb^{(n)} \to \rb$ as $n \to \infty$. So difference vector $\rhob^{(n)} \defeq \rb^{(n)} - \rb$ converges to $0$. Set $\alb^{(n)} \defeq M \rhob^{(n)}$. Then $\| \alb^{(n)} \| = O \big( \| \rhob^{(n)} \| \big)$ converges to $0$ as well.

With these notations, we compute the coordinates of the next point $\rb^{(n+1)} = \Fr \big( \rb^{(n)} \big)$ of our sequence: 
\[ r^{(n+1)}_j = \bigg( \sum_{x \in j} \frac{p_x}{a_x+\al^{(n)}_x} \bigg) ( r_j + \varrho^{(n)}_j ) 
= \bigg( \underbrace{\sum_{x \in j} \frac{p_x}{a_x}}_{=1} 
- \sum_{x \in j} \frac{p_x \al^{(n)}_x }{a^2_x} 
+ \sum_{x \in j} \frac{p_x \big( \al^{(n)}_x \big)^2 }{a^2_x(a_x+\al^{(n)}_x)}
\bigg) ( r_j + \varrho^{(n)}_j ).\]
It follows that 
\[ \varrho^{(n+1)}_j = \varrho^{(n)}_j - r_j \sum_{x \in j} \frac{p_x}{a_x^2} \al^{(n)}_x + O\big( \| \rhob^{(n)} \|^2 \big).\]
Since $\alb^{(n)} = M \rhob^{(n)}$, we conclude that 
\begin{equation} \label{eq:diff_vector_step}
\rhob^{(n+1)} = (I-N) \rhob^{(n)} + O\big( \| \rhob^{(n)} \|^2 \big) .
\end{equation}
Recall that $\lam$ is the smallest nonzero eigenvalue of $N$. Under our assumption $\ker N = \ker M = \{0\}$, so $0$ is not an eigenvalue now, meaning that the largest eigenvalue of the diagonalizable matrix $I-N$ is $1-\lam$. Then it is not hard to deduce from \eqref{eq:diff_vector_step} that 
\[ \| \rhob^{(n)} \| = O\big( (1-\lam)^n \big) .\]
%
As for the $\phir$-value, 
\[ \phir\big( \rb^{(n)} \big) = \phir\big( \rb + \rhob^{(n)} \big) 
= \phir(\rb) + \underbrace{\nabla \phir (\rb)}_{=-\eb} \cdot \rhob^{(n)} + O\big( \| \rhob^{(n)} \|^2 \big) ,\]
where the dot product $\eb \cdot \rhob^{(n)}$, which is simply the sum of the coordinates of $\rhob^{(n)}$, is equal to $0$ because this sum is $1$ both for $\rb \in \Kr$ and for $\rb^{(n)} \in \Kr$. Then \eqref{eq:conv_rate} clearly follows.

In fact, heuristically, 
$\rhob^{(n+1)} \approx (I-N) \rhob^{(n)}$ means that if we write $\rhob^{(n)}$ in an eigenbasis, then the parts corresponding to smaller eigenvalues will become negligible and $\rhob^{(n)}$ will be close to an eigenvector with the maximal eigenvalue $1-\lam$, and hence $\rhob^{(n+1)} \approx (1-\lam) \rhob^{(n)}$ for large $n$ (at least for typical starting points). 

When $\ker M = \ker N$ has positive dimension, $1$ is an eigenvalue of $I-N$ and it seems that we do not necessarily have exponential decay. Note, however, that vectors from $\ker M$ do not make a difference from the point of view of $\varphi$-value because for any $\vb \in \ker M$ we have $A( \rprb + \vb) = A(\rprb)$, and hence $\phir( \rprb + \vb) = \phir(\rprb)$.

We close this section by proving the required properties of $N$.
\begin{proof}[Proof of Lemma \ref{lem:matrix_N}]
Since $\rb$ is a minimum point of $\phir$, by Proposition \ref{prop:min_fixed} $\rb$ is a fixed point of $\Fr$, and hence $\Delta_j(\rb)=1$ for each $j$. This means that $M^\top (\pb \ab^{-1})$ is the all-ones vector $\eb$. 

We also have $M \rb = A(\rb) = \ab$. Then 
\[ N^\top \eb = M^\top D(\pb \ab^{-2}) M \underbrace{D(\rb) \eb}_{=\rb} 
= M^\top D(\pb \ab^{-2}) \underbrace{M \rb}_{=\ab} 
= M^\top \underbrace{D(\pb \ab^{-2}) \ab}_{=\pb \ab^{-1}} = M^\top (\pb \ab^{-1}) = \eb ,\]
confirming that $1$ is an eigenvalue and that each column sum of $N$ is $1$. Since all entries are nonnegative, it follows that $N$ is a (left) stochastic matrix, and hence $|\lambda| \leq 1$ for each eigenvalue $\la$.

Furthermore, $N$ is similar to a positive semidefinite matrix:
\[ D(\rb^{-1/2}) N D(\rb^{1/2}) = D(\rb^{1/2}) M^\top D(\pb \ab^{-2}) M D(\rb^{1/2}) ,\]
so $N$ is diagonalizable with nonnegative eigenvalues.

Finally, let $B = D(\pb^{1/2} \ab^{-1}) M$. Then 
\[ \ker M = \ker B = \ker(B^\top B) = \ker( M^\top D(\pb \ab^{-2}) M ) = \ker N .\]
\end{proof}

\subsubsection*{Convergence for a redundant set}
If $r_j=0$ for a given $j$ at the limiting point $\rb = \lim_{n \to \infty} \rb^{(n)}$, then we must have $\Delta_j( \rb ) < 1$. We then eventually see an exponential decay in the $j$-coordinate:
\[ r^{(n+1)}_j = \big( \Delta_j( \rb ) + o(1) \big) r^{(n)}_j .\]
Then the growth rate for the precision of $\phir\big( \rb^{(n)} \big)$ is at most $-\log \Delta$, where $\Delta$ denotes the largest value of $\Delta_j( \rb )$ among all redundant sets $j$. For the dodecahedral graph we have $\Delta=7/8$, which is consistent with our previous numerical findings.

\section{Conclusion}

The optimal rate of lossless functional compression with side information at the receiver can be characterized by conditional graph entropy. However, little can be found in the literature about this entropy notion. So we set out to study conditional graph entropy in more detail. Our starting point was the original formula \eqref{eq:mut_inf_cond} which can also be formulated as the q-problem. Our first step was the discovery of the r-problem and the stepping maps $Q, R$ between the two problems. This interaction was reminiscent of the alternating optimization in the EM algorithm, which made us realize that there might be an underlying alternating minimization problem. This, in turn, helped us to analyze the iterative algorithm because we could turn to the general theory of Csisz\'ar and Tusn\'ady: we verified that the 3-point, 4-point, and 5-point properties hold in our setting, and showed that the iterations always converge to the minimum. Our theoretical results lead to a practical algorithm for computing conditional graph entropy that also comes with an error bound based on a dual problem.

Alternating optimization has a vast and growing literature. The fact that (conditional) graph entropy is part of this family of problems will hopefully inspire future research in the area.

\section*{Acknowledgments}
We are grateful to two anonymous reviewers for their numerous valuable remarks and suggestions that helped us tremendously in improving the paper.

\bibliographystyle{plain}
\bibliography{refs}

\begin{thebibliography}{10}

\bibitem{Alon1996}
N.~Alon and A.~Orlitsky.
\newblock Source coding and graph entropies.
\newblock {\em IEEE Transactions on Information Theory}, 42(5):1329--1339,
  1996.

\bibitem{Arimoto1972}
S.~Arimoto.
\newblock An algorithm for computing the capacity of arbitrary discrete
  memoryless channels.
\newblock {\em IEEE Transactions on Information Theory}, 18(1):14--20, 1972.

\bibitem{Blahut1972}
R.~Blahut.
\newblock Computation of channel capacity and rate-distortion functions.
\newblock {\em IEEE Transactions on Information Theory}, 18(4):460--473, 1972.

\bibitem{boreland_thesis}
Gareth Boreland.
\newblock {\em Information theoretic parameters for graphs and operator
  systems}.
\newblock PhD thesis, Queen’s University Belfast, 2020.

\bibitem{entropy_splitting}
Imre Csisz\'ar, J\'anos K\"orner, L\'aszl\'o Lov\'asz, Katalin Marton, and
  G\'abor Simonyi.
\newblock Entropy splitting for antiblocking corners and perfect graphs.
\newblock {\em Combinatorica}, 10(1):27--40, 1990.

\bibitem{csiszar_tusnady}
Imre Csisz{\'a}r and G{\'a}bor Tusn{\'a}dy.
\newblock Information geometry and alternating minimization procedures.
\newblock {\em Statistics and Decisions}, Supp.~1:205–237, 1984.

\bibitem{doshi_et_al}
Vishal Doshi, Devavrat Shah, Muriel Médard, and Michelle Effros.
\newblock Functional compression through graph coloring.
\newblock {\em IEEE Transactions on Information Theory}, 56(8):3901--3917,
  2010.

\bibitem{Dupuis2004}
F.~Dupuis, W.~Yu, and F.M.J. Willems.
\newblock Blahut-arimoto algorithms for computing channel capacity and
  rate-distortion with side information.
\newblock In {\em International Symposium on Information Theory, 2004. ISIT
  2004. Proceedings.}, pages 179--, 2004.

\bibitem{GEL1980}
S.I. Gel'fand and M.S. Pinsker.
\newblock Coding for channel with random parameters.
\newblock {\em Probl. Contr. and Inf. Theory}, 1980.

\bibitem{grotschel1986}
M.~Gr\"{o}tschel, L.~Lov\'{a}sz, and A.~Schrijver.
\newblock Relaxations of vertex packing.
\newblock {\em J. Combin. Theory Ser. B}, 40(3):330--343, 1986.

\bibitem{ge_github}
Viktor Harangi.
\newblock Graph entropy program code.
\newblock \url{https://github.com/harangi/graphentropy}, 2023.

\bibitem{korner1973}
J{\'a}nos K{\"o}rner.
\newblock Coding of an information source having ambiguous alphabet and the
  entropy of graphs.
\newblock In {\em 6th Prague conference on information theory}, pages 411--425,
  1973.

\bibitem{lovasz1979}
L\'{a}szl\'{o} Lov\'{a}sz.
\newblock On the {S}hannon capacity of a graph.
\newblock {\em IEEE Trans. Inform. Theory}, 25(1):1--7, 1979.

\bibitem{orlitsky_roche}
Alon Orlitsky and James~R. Roche.
\newblock Coding for computing.
\newblock {\em IEEE Transactions on Information Theory}, 47:903--917, 1998.

\bibitem{ShannonWeaver49}
Claude~E. Shannon and Warren Weaver.
\newblock {\em The Mathematical Theory of Communication}.
\newblock University of Illinois Press, Urbana and Chicago, 1949.

\bibitem{survey}
G{\'a}bor Simonyi.
\newblock Graph entropy: A survey.
\newblock In William Cook, L{\'a}szl{\'o} Lov{\'a}sz, and Paul Seymour,
  editors, {\em Combinatorial Optimization}, volume~20 of {\em DIMACS Series in
  Discrete Mathematics and Theoretical Computer Science}, pages 399--441, 1993.

\bibitem{survey2}
G{\'a}bor Simonyi.
\newblock Perfect graphs and graph entropy. {A}n updated survey.
\newblock In Jorge Ramirez-Alfonsin and Bruce Reed, editors, {\em Perfect
  Graphs}, pages 293--328. John Wiley and Sons, 2001.

\bibitem{Vontobel2008}
Pascal~O. Vontobel, Aleksandar Kavcic, Dieter~M. Arnold, and Hans-Andrea
  Loeliger.
\newblock A generalization of the {B}lahut–{A}rimoto algorithm to
  finite-state channels.
\newblock {\em IEEE Transactions on Information Theory}, 54(5):1887--1918,
  2008.

\bibitem{Vrana2021}
P\'{e}ter Vrana.
\newblock Probabilistic refinement of the asymptotic spectrum of graphs.
\newblock {\em Combinatorica}, 41(6):873--904, 2021.

\bibitem{Witsenhausen1976}
H.~Witsenhausen.
\newblock The zero-error side information problem and chromatic numbers
  (corresp.).
\newblock {\em IEEE Transactions on Information Theory}, 22(5):592--593, 1976.

\end{thebibliography}

\end{document}